\documentclass[conference]{IEEEtran}

\usepackage{cite}

\ifCLASSINFOpdf
\else
\fi
\usepackage{amsmath}
\usepackage{array}
\usepackage{url}

\usepackage[dvipsnames]{xcolor}

\newif\ifdraft\drafttrue

\ifdraft
\newcommand\modedraft[1]{#1}
\newcommand\todo[1]{{\color{purple}[\textbf{To do:} #1]}}
\newcommand\bmcomment[1]{{\footnotesize \color{blue}[#1 -
    \textbf{Bastien}]}}
\newcommand\amcomment[1]{{\footnotesize \color{OliveGreen}[#1 - \textbf{Nello}]}}
\newcommand\rbcomment[1]{{\footnotesize \color{Orange}[#1 - \textbf{Raphael}]}}

\newcommand\sr[1]{{\color{Red}{[SR: #1]}}}
\else
\newcommand\modedraft[1]{}
\newcommand\todo[1]{}
\newcommand\bmcomment[1]{}
\newcommand\amcomment[1]{}
\newcommand\rbcomment[1]{}

\newcommand\sr[1]{}
\fi

\usepackage{amssymb,stmaryrd} 
\usepackage{amsthm}
\usepackage{calc} 
\usepackage{xspace}
\usepackage{macros,sectionningmacros}
\usepackage{bm} 
\usepackage{tikz}
\usepackage{changebar}
\usepackage{hyperref}

\IEEEoverridecommandlockouts

\begin{document}
%
\title{Strategy Logic with Imperfect Information}

\author{\IEEEauthorblockN{Rapha\"el Berthon\IEEEauthorrefmark{1},
    Bastien Maubert\IEEEauthorrefmark{2}\thanks{This project has
      received funding from the European Union's Horizon 2020 research
      and innovation programme under the Marie Sklodowska-Curie grant
      agreement No
      709188. 
      It was also supported in part by NSF grants
      CCF-1319459 and IIS-1527668, and by NSF Expeditions in Computing
      project "ExCAPE: Expeditions in Computer Augmented Program
      Engineering". Finally, this work was done while Sasha Rubin was
      a Marie Curie fellow of the Istituto Nazionale di Alta
      Matematica.}, Aniello Murano\IEEEauthorrefmark{2}, Sasha
    Rubin\IEEEauthorrefmark{2} and Moshe
    Y. Vardi\IEEEauthorrefmark{3}}
  \IEEEauthorblockA{\IEEEauthorrefmark{1}\'Ecole Normale Sup\'erieure
    de Rennes, Rennes, France}
  \IEEEauthorblockA{\IEEEauthorrefmark{2}Universit\`a degli Studi di
    Napoli Federico II, Naples, Italy}
  \IEEEauthorblockA{\IEEEauthorrefmark{3}Rice University, Houston,
    Texas, USA } }




%


\newcommand\copyrighttext{%
  \footnotesize \textcopyright 2017 IEEE. Personal use of this material is permitted.
  Permission from IEEE must be obtained for all other uses, in any current or future
  media, including reprinting/republishing this material for advertising or promotional
  purposes, creating new collective works, for resale or redistribution to servers or
  lists, or reuse of any copyrighted component of this work in other works.
  DOI: \href{https://ieeexplore.ieee.org/document/8005136/}{10.1109/LICS.2017.8005136}}
\newcommand\copyrightnotice{%
\begin{tikzpicture}[remember picture,overlay]
\node[anchor=south,yshift=10pt] at (current page.south) {\fbox{\parbox{\dimexpr\textwidth-\fboxsep-\fboxrule\relax}{\copyrighttext}}};
\end{tikzpicture}%
}


\newcommand\UElogo{%
\begin{tikzpicture}[remember picture,overlay]
\node[anchor=south,yshift=45pt,xshift=0cm] at (current page.south) {\includegraphics[height=2.5em]{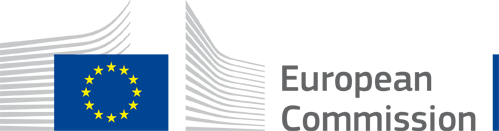}};
\end{tikzpicture}%
}



\maketitle
\copyrightnotice
\UElogo

\begin{abstract}

We introduce an extension of Strategy logic for the imperfect-information 
setting, called \SLi, and study its model-checking problem. As this logic
naturally captures multi-player games with imperfect information, the problem
turns out to be undecidable. We introduce a syntactical class of ``hierarchical
instances'' for which, intuitively, as one goes down the syntactic tree of the
formula, strategy quantifications are concerned with finer observations of the
model. We prove that model-checking \SLi restricted to hierarchical instances
is decidable.  This result, because it allows for complex patterns of
existential and universal quantification on strategies, greatly generalises
previous ones, such as decidability of multi-player games with 
imperfect information and hierarchical observations, and decidability of distributed
synthesis for hierarchical systems.

To establish the decidability result, we introduce and study \QCTLsi, an
extension of \QCTL (itself an extension of \CTL with second-order quantification
over atomic propositions) by parameterising its quantifiers with observations.
The simple syntax of \QCTLsi allows us to provide a conceptually neat reduction
of  \SLi to \QCTLsi  that separates concerns, allowing one to forget about
strategies and players and focus solely on second-order quantification. While the
model-checking problem of \QCTLsi is, in general, undecidable, we identify a
syntactic fragment of hierarchical formulas and prove, using an
automata-theoretic approach, that it is decidable. The decidability result for
\SLi follows since the reduction maps hierarchical instances of \SLi to
hierarchical formulas of \QCTLsi.

\end{abstract}



%
\IEEEpeerreviewmaketitle

\section{Introduction}


Logics for strategic reasoning are a powerful tool for expressing
correctness properties of multi-player graph-games, which in turn are natural
models for reactive systems and discrete event systems.
In particular, \ATLs and its related
logics were introduced to capture the realisability/synthesis problem
for open systems with multiple components. These logics were designed
as extensions of branching-time logics such as \CTLs that allow one to
write alternating properties directly in the syntax, \ie, statements
of the form ``there exist strategies, one for each player in $A$, such
that for all 
strategies of the remaining players, the resulting play satisfies $\phi$''. 
Strategy logic (\SL)~\cite{DBLP:journals/tocl/MogaveroMPV14}
generalises these by treating
strategies as first-order objects $x$ that can be quantified 
$\Estrato{}$ (read ``there exists a strategy $x$'') and bound 
to players $(a,x)$ (read ``player $a$ uses strategy $x$''). 
This syntax has flexibility
very similar to first-order logic, and thus allows one to directly
express many solution concepts from game-theory, e.g., the \SL formula
$\Estrat[x_1] \Estrat[x_2] (a_1,x_1) (a_2,x_2) \wedge_{i=1,2} [ \Estrat[y_i] (a_i,y_i) goal_i ] \to goal_i$ expresses 
the existence of a Nash equilibrium in a two-player game (with individual Boolean objectives).

An essential property of realistic multi-player games is that players
only have a limited view of the state of the system. This is captured 
by introducing partial-observability into the models, \ie, equivalence-relations
$\obs$ (called \emph{observations}) over the state space that specify
indistinguishable states. In the formal-methods literature it is typical to associate
observations to players. In this paper, instead, we associate observations 
to strategies. Concretely, we introduce an extension \SLi of
\SL that annotates the strategy quantifier $\Estrat{}$ by an
observation $\obs$, written $\Estrato{\obs}$. Thus, both the model and
the formulas mention observations $\obs$. This novelty allows one to
express, in the logic, that a player's observation changes over
time.


Our logic \SLi is extremely powerful: it is an extension of \SL (which
is in the perfect-information setting), and of the
imperfect-information strategic logics \ATLsi~\cite{BJ14} and
\ATLssci~\cite{DBLP:journals/corr/LaroussinieMS15}.  A canonical
specification in multi-player game of partial observation is that the
players, say $a_1, \dots, a_n$, can form a coalition and beat the
environment player, say $b$. This can be expressed in \SLi as
$\Phi_{\textsc{Synth}} \egdef \Estrato[x_1]{\obs_1} \dots
\Estrato[x_n]{\obs_n} \Astrato[y]{\obs} (a_1, x_1) \dots (a_n, x_n)
(b,y) \phi$, where $\phi$ is quantifier- and binding-free. Also, \SLi
can express more complicated specifications by alternating
quantifiers, binding the same strategy to different agents and
rebinding (these are inherited from \SL), as well as changing
observations.


The complexity of \SLi is also visible from an algorithmic point of
view.  Its satisfiability problem is undecidable (this is already true
of \SL), and its model-checking problem is undecidable (this is
already true of \ATLsi, in fact, even for the single formula
$\EstratATL[\{a,b\}] \F p$ in systems with three
players~\cite{CT11}). In fact, similar undecidability occurs in
foundational work in multi-player games of partial observation, and in
distributed synthesis~\cite{PR90,DBLP:conf/focs/PetersonR79}. Since
then, the formal-methods community has spent much effort finding
restrictions and variations that ensure decidability~\cite{
  kupermann2001synthesizing,
  peterson2002decision,DBLP:conf/lics/FinkbeinerS05,pinchinat2005decidable,DBLP:conf/atva/ScheweF07,DBLP:conf/icalp/Chatterjee014,
BMM17,BLMR17,DBLP:conf/atva/BerwangerMB15}.  The common thread in these
approaches is that the players' observations (or what they can deduce
from their observations) are hierarchical.

Motivated by the problem of finding decidable extensions of strategy logic in the imperfect-information setting,
we introduce a syntactic class of ``hierarchical instances'' of \SLi,
\ie, formula/model pairs, and prove that the model-checking problem
on this class of instances is decidable. Intuitively, an instance of \SLi is hierarchical if, as one goes
down the syntactic tree of the formula, the observations annotating
strategy quantifications can only become finer. Although the class of
hierarchical instances refers not only to the syntax of the logic but
also to the model, the class is syntactical in the sense that it depends
only on the structure of the formula and the observations in the
model. Moreover, it is easy to check (\ie, in linear time) if an
instance is hierarchical or not.
 
The class of hierarchical instances generalises some existing
approaches and supplies new classes of systems and properties that can
be model-checked. For instance, suppose that there is a total order
$\preceq$ among the players such that $a \preceq b$ implies player
$b$'s observation is finer than player $a$'s observation --- such
games are said to yield ``hierarchical observation'' in~\cite{DBLP:conf/atva/BerwangerMB15}. 
In such games it is known that
synthesis against \CTLs specifications is decidable (in fact, this
holds for $\omega$-regular
specifications~\cite{peterson2002decision,DBLP:conf/atva/BerwangerMB15}).
This corresponds to hierarchical instances of \SLi in which the
observations form a total order in the model and the formula is of the
form $\Phi_{\textsc{Synth}}$ (mentioned above).  On the other hand, in
hierarchical instances of \SLi, the ordering on observations can be a
pre partial-order (\ie, not just total), and one can arbitrarily
alternate quantifiers in the formulas.  For instance, hierarchical
instances allow one to decide if a game that yields hierarchical
information has a Nash equilibrium. 
For example, assuming observations $p_1,p_2,o_1,o_2$ with
$p_i$ finer than $o_2$ (for $i = 1,2$), and $o_2$ finer than $o_1$,
the formula $\Estrato[x_1]{\obs_1} \Estrato[x_2]{o_2} (a_1,x_1)
(a_2,x_2) \wedge_{i=1,2} [ \Estrato[y_i]{p_i} (a_i,y_i) goal_i ] \to
goal_i$ expresses that there exists a strategy profile $(x_1,x_2)$ of
uniform strategies, such that neither player has an incentive to
deviate using a strategy that observes at least as much as the
observations that both players started with. Observe that this formula
is in fact \emph{equivalent} to the existence of a Nash equilibrium,
\ie, to the same formula in which $p_i = o_i$. This shows we can
decide the existence of Nash equilibria in games that yield
hierarchical observation.


As a tool to study \SLi we introduce 
\QCTLsi, an extension to the imperfect-information 
setting of \QCTLs~\cite{DBLP:journals/corr/LaroussinieM14}, 
itself an extension of \CTLs by  second-order
quantifiers over atoms.  This is a low-level logic that does not
mention strategies and into which one can effectively compile
instances of \SLi.  States of the models of the logic \QCTLsi have
internal structure, much like the multi-player game structures from~\cite{peterson2001lower} and distributed
systems~\cite{halpern1989complexity}. Model-checking \QCTLsi is also
undecidable (indeed, we show how to reduce from the MSO-theory of the
binary tree extended with the equal-length predicate, known to be
undecidable~\cite{lauchli1987monadic}).  We introduce the syntactical
class \QCTLsih of hierarchical formulas as those in which innermost
quantifiers observe more than outermost quantifiers, and prove that
model-checking is decidable using an extension of the
automata-theoretic approach for branching-time logics (our decision to
base models of \QCTLsi on local states greatly eases the use of
automata). Moreover, the compilation of \SLi into \QCTLsi preserves
being hierarchical, thus  establishing our main
contribution, \ie, that model checking the hierarchical instances of
\SLi is decidable.


\head{Related work} Formal methods for reasoning about reactive
systems with multiple components have been studied mainly in two
theoretical frameworks: a) multi-player graph-games of
partial-observation~\cite{DBLP:conf/focs/PetersonR79,peterson2002decision,DBLP:conf/atva/BerwangerMB15}
and b) synthesis in distributed
architectures~\cite{PR90,kupermann2001synthesizing,DBLP:conf/lics/FinkbeinerS05,DBLP:conf/atva/ScheweF07,DBLP:journals/fmsd/GastinSZ09}
(the relationship between these two frameworks is discussed
in~\cite{DBLP:conf/atva/BerwangerMB15}).  All of these works consider
the problem of synthesis, which (for objectives in temporal logics)
can be expressed in \SLi using the formula $\Phi_{\textsc{Synth}}$
mentioned above. Limited alternation was studied in
\cite{DBLP:conf/icalp/Chatterjee014} that, in the language of \SLi,
considers the model-checking problem of formulas of the form
$\Estrato[x_1]{\obs_1} \Astrato[x_2]{\obs_2} \Estrato[x_3]{\obs_3}
(a_1, x_1) (a_2,x_2)(a_3, x_3) \phi$, where $\phi$ is an
$\omega$-regular objective. They prove that this is decidable in case
 player $3$ has perfect observation and player $2$ observes at
least as much as player $1$.

In contrast to all these works, formulas of \SLi can express much more complex specifications by alternating quantifiers, sharing strategies, rebinding, and changing observations.  

Recently, \cite{DBLP:conf/atva/BerwangerMB15} generalised the classic
result of \cite{peterson2002decision}: it weakens the assumption of
hierarchical observation to hierarchical information (which are both
static notions), and then, further to dynamic hierarchical information
which allows for the hierarchy amongst players' information to change
along a play. However, it only considers the synthesis problem. 

We are aware of two papers that (like we do) give simultaneous structural constraints on both the formula and the model that result in decidability: in the context of synthesis in distributed architecture with process delays, \cite{DBLP:journals/fmsd/GastinSZ09} considers \CTLs specifications that constrain external variables by the input variables that may effect them in the architecture; and in the context of asynchronous perfect-recall, \cite{pinchinat2005decidable} considers a syntactical restriction on instances for Quantified $\mu$-Calculus with partial observation (in contrast, we consider the case of synchronous perfect recall).

The work closest to ours is~\cite{DBLP:conf/csl/FinkbeinerS10} which introduces a decidable logic $\CL$ in which one can encode many distributed synthesis problems. However, \CL is close in spirit to our \QCTLsih, and is more appropriate as a tool than as a high-level specification logic like \SLi. Furthermore, by means of a natural translation we derive that \CL is strictly included in the hierarchical instances of \SLi (Section~\ref{subsec:SLi-comparison}).  
In particular, we find that hierarchical instances of \SLi can express non-observable goals, while \CL does not. Non-observable goals arise naturally in problems in distributed synthesis~\cite{PR90}.

Finally, our logic \SLi is the first generalisation of 
\SL to include strategies with partial observation and,
{unlike \CL}, to generalise previous logics with
partial-observation strategies, \ie, \ATLsi~\cite{BJ14} and
\ATLssci~\cite{DBLP:journals/corr/LaroussinieMS15}. A comparison of
\SLi to \SL, \ATLsi, \ATLssci and \CL is given in
Section~\ref{subsec:SLi-comparison}.

\head{Outline}
The definition of \SLi and of hierarchical instances, and the discussion about Nash equilibria, are in Section~\ref{sec-SLi}. The definition of \QCTLsi, and its hierarchical fragment \QCTLsih, are in Section~\ref{sec-QCTL-imp-inf}.
The proof that model checking \QCTLsih is decidable, including the required automata preliminaries, are in Section~\ref{sec-decidable}.
The translation of \SLi into \QCTLsi, and the fact that this preserves
hierarchy, are in Section~\ref{sec-modelcheckingSL}.



\section{\SL with imperfect information}
\label{sec-SLi}

In this section we introduce \SLi, an extension of \SL~\cite{DBLP:journals/tocl/MogaveroMPV14} 
to the imperfect-information setting with synchronous perfect-recall. Our logic presents two original features: first, observations are not
bound to players (as is done in extensions of \ATL by
imperfect information~\cite{DBLP:journals/jacm/AlurHK02} or 
logics for reasoning about knowledge~\cite{fagin1995reasoning}), and second, 
we have syntactic observations in the language, which need to be interpreted. 

We follow the presentation of \SL in \cite{DBLP:journals/tocl/MogaveroMPV14}, except that
we make some changes that simplify their presentation
but do not change their semantics, and some that allow us to capture
imperfect information. We introduce the syntax and semantics of \SLi,
carefully detailing these changes. We first fix some notation used throughout the paper.

\subsection{Notation}
Let $\Sigma$ be an alphabet. A \emph{finite} (resp. \emph{infinite}) \emph{word} over $\Sigma$ is an element
of $\Sigma^{*}$ (resp. $\Sigma^{\omega}$). Words are written $w = w_0 w_1 w_2 \ldots$, \ie, indexing begins with $0$. 
The \emph{length} of a finite word $w=w_{0}w_{1}\ldots
w_{n}$ is $|w|\egdef n+1$, and $\last(w)\egdef w_{n}$ is its last
letter.
Given a finite (resp. infinite) word $w$ and $0 \leq i \leq |w|$ (resp. $i\in\setn$), we let $w_{i}$ be the
letter at position $i$ in $w$, $w_{\leq i}$ is the prefix of $w$ that
ends at position $i$ and $w_{\geq i}$ is the suffix of $w$ that starts
at position $i$.
We write $w\pref w'$ if $w$ is a prefix of $w'$, and $\FPref{w}$ is
the set of finite prefixes of word $w$.
Finally, 
the domain of a mapping $f$ is written $\dom(f)$,  and for $n\in\setn$ we let $[n]\egdef\{i \in \setn: 1 \leq i \leq n\}$.
The literature sometimes refers to ``imperfect information'' and sometimes to ``partial observation''; we will use the terms interchangeably.
  
\subsection{Syntax}
\label{sec-SLi-definition}
 
The syntax of \SLi is similar to that of strategy logic \SL as defined
in \cite{DBLP:journals/tocl/MogaveroMPV14}: the only difference is
that we annotate strategy quantifiers $\Estrato{}$ by
\emph{observation symbols} $\obs$.  For the rest of the paper, for
convenience we fix a number of parameters for our logics and models:
$\APf$ is a finite set of \emph{atomic propositions}, $\Agf$ is a
finite set of \emph{agents} or \emph{players}, $\Varf$ is a finite set
of \emph{variables} and $\Obsf$ is a finite set of \emph{observation
  symbols}.  When we consider model-checking problems, these data are
implicitly part of the input.
%

\begin{definition}[\SLi Syntax]
  \label{def-SLi}
  The syntax of \SLi is defined by the following grammar:
  \[
  \phi \egdef p 
  \mid \neg \phi 
  \mid \phi\ou\phi 
  \mid \X \phi 
  \mid \phi \until \phi 
  \mid \Estrato{\obs}\phi 
  \mid \bind{\var}\phi
  \]
 where 
  $p\in\APf$, $\var\in\Varf$, $\obs\in\Obsf$ and $a\in\Agf$.
\end{definition}

We use  abbreviations $\top\egdef p\ou\neg p$, $\perp\egdef\neg\top$, $\phi\impl\phi'\egdef \neg \phi \ou \phi'$,
$\phi\equivaut\phi'\egdef \phi\impl\phi'\et \phi'\impl\phi$ for
boolean connectives,
 $\F\phi \egdef \top \until \phi$,  $\G\phi \egdef \neg \F
\neg \phi$ for temporal operators, and finally $\Astrato{\obs}\phi\egdef\neg\Estrato{\obs}\neg\phi$.

The notion of free variables and sentences are defined as for \SL. 
We recall these for completeness:
A variable $\var$ appears \emph{free} in a formula $\phi$ if it
appears out of the scope of a strategy quantifier, and a
player $\ag$ appears free in $\phi$ if a temporal operator (either $\X$ or
$\until$) appears in $\phi$ out of the scope of any binding for player $\ag$. 
We let $\free(\phi)$ be the set of variables and players  that appear
free in $\phi$. If $\free(\phi)$ is empty, $\phi$ is a \emph{sentence}.

\subsection{Semantics}
\label{sec-SLmodels}

The models of \SLi are like those of \SL, \ie, concurrent game structures, but
extended by a finite set of observations $\Obsf$ and, for each $o \in \Obsf$, 
by an equivalence-relation $\obsint(o)$ over positions that represents what a player
using a strategy with that observation can see. That is, $\obsint(o)$-equivalent positions
are indistinguishable to a player using a strategy associated with observation
$o$.


\begin{definition}[\CGSi]
  \label{def-CGSi}
  A \emph{concurrent game structure with imperfect information} (or
  \CGSi for short)
is a structure  $\CGSi=(\Act,\setpos,\trans,\val,\pos_\init,\obsint)$ where
   \begin{itemize}
    \item $\Act$ is a finite non-empty set of \emph{actions},
    \item $\setpos$ is a finite non-empty set of \emph{positions},
   \item $\trans:\setpos\times \Mov^{\Agf}\to \setpos$ is a \emph{transition function}, 
  \item $\val:\setpos\to 2^{\APf}$ is a \emph{labelling function},  
  \item $\pos_\init \in \setpos$ is an \emph{initial position},
  \item $\obsint:\Obsf\to 2^{\setpos\times\setpos}$ is an 
  \emph{observation interpretation} mapping observations 
  to equivalence relations on positions.
 \end{itemize}
\end{definition}

We may  write $\obseq$ for $\obsint(\obs)$, and $\pos\in\CGS$ for $\pos\in\setpos$.

The notions of joint actions, plays, strategies and assignments
are similar to those for \SL. We will recall these for completeness.
Since the main difference between the
semantics of \SL and \SLi is in the strategy-quantification case where
we require strategies to be consistent with observations, we define
synchronous perfect recall and $\obs$-strategies before defining the
semantics of \SLi. Finally, we mention that we make some
(inconsequential) simplifications to the  semantics of \SL
as proposed in \cite{DBLP:journals/tocl/MogaveroMPV14}: i) we dispense
with partial assignments and only consider complete assignments, ii)
our semantics are w.r.t. a finite play instead of a position (these simplifications do not change the
expressive power of \SL).


\head{Joint actions}
In a position $\pos\in\setpos$, each player $\ag$ chooses an action $\mova\in\Mov$, 
and the game proceeds to position
$\trans(\pos, \jmov)$, where $\jmov\in \Mov^{\Agf}$ stands for the \emph{joint action}
$(\mova)_{\ag\in\Agf}$. Given a joint action
$\jmov=(\mova)_{\ag\in\Agf}$ and $\ag\in\Agf$, we let
$\jmov_{\ag}$ denote $\mova$. 
For each position $\pos\in\setpos$, $\val(\pos)$ is the set of atomic
propositions that hold in $\pos$.

\head{Plays and strategies}
A \emph{finite} (resp. \emph{infinite}) \emph{play} is a finite (resp. infinite)
word $\fplay=\pos_{0}\ldots \pos_{n}$ (resp. $\iplay=\pos_{0} \pos_{1}\ldots$)
such that $\pos_{0}=\pos_{\init}$ and for all $i$ with $0\leq i<|\fplay|-1$ (resp. $i\geq 0$), there exists a joint action $\jmov$
such that $\trans(\pos_{i}, \jmov)=\pos_{i+1}$. 
We let $\FPlay$ be the set of finite plays. 
A \emph{strategy} is a function $\strat:\FPlay \to \Mov$, and the set of all strategies
is denoted $\setstrat$.

\head{Assignments} 
An \emph{assignment} is a function $\assign:\Agf\union\Varf \to \setstrat$, assigning a strategy to
each player and variable. For an assignment
$\assign$,  player $a$ and  strategy $\strat$,
$\assign[a\mapsto\strat]$ is the assignment that maps $a$ to $\strat$ and is equal to
$\assign$ on the rest of its domain, and 
$\assign[\var\mapsto \strat]$ is defined similarly, where $\var$ is a variable. 

\head{Outcomes}
For an assignment $\assign$ and a finite play $\fplay$, we let
$\out(\assign,\fplay)$ be the only infinite play that starts with
$\fplay$ and is then extended by letting players follow the strategies
assigned by $\assign$. Formally,
 $\play(\assign,\fplay)\egdef \fplay \cdot \pos_{1}\pos_{2}\ldots$ where,
 for all $i\geq 0$,
 $\pos_{i+1}=\trans(\pos_{i},\jmov)$, where
 $\pos_{0}=\last(\fplay)$ and
 $\jmov=(\assign(a)(\fplay\cdot\pos_{1}\ldots\pos_i))_{a\in\Agf}$.

\head{Synchronous perfect recall}
In this work we consider
players with \emph{synchronous perfect recall}, meaning that each player
remembers the whole history of a play, a classic assumption in games
with imperfect information and logics of knowledge and time. Each observation
 relation is thus extended to finite plays
as follows: $\fplay \obseq \fplay'$ if $|\fplay|=|\fplay'|$
and $\fplay_{i}\obseq\fplay'_{i}$ for every $i\in\{0,\ldots,
|\fplay|-1\}$.  For $\obs\in\Obsf$, an \emph{$\obs$-strategy} is a strategy $\strat:\setpos^{+} \to \Mov$ 
such that $\strat(\fplay)=\strat(\fplay')$ whenever $\fplay \obseq
\fplay'$. The latter  constraint captures the essence of
imperfect information, which is that players can base their strategic
choices only on the information available to them. 
For $\obs\in\Obsf$ we let $\setstrato$ be the set of
all $\obs$-strategies. 

\begin{definition}[\SLi Semantics]
\label{def-SL-semantics}
The semantics $\CGSi,\assign,\fplay\modelsSL \phi$ is defined inductively, where
$\phi$ is an \SLi-formula,
$\CGSi=(\Act,\setpos,\trans,\val,\pos_\init,\obsint)$ is a \CGSi, 
$\fplay$ is a finite play, and 
$\assign$ is an assignment:
\[
\begin{array}{l}
\begin{array}{lcl}
 \CGSi,\assign,\fplay\modelsSL p & \text{ if } & p\in\val(\last(\fplay))\\
 \CGSi,\assign,\fplay\modelsSL \neg\phi & \text{ if } &
  \CGSi,\assign,\fplay\not\modelsSL\phi\\
 \CGSi,\assign,\fplay\modelsSL \phi\ou\phi' & \text{ if } &
  \CGSi,\assign,\fplay\modelsSL\phi \text{ or }
  \CGSi,\assign,\fplay\modelsSL\phi' \\
 \CGSi,\assign,\fplay\modelsSL\Estrato{\obs}\phi  & \text{ if } & 
\exists\,   \strat\in\setstrato \text{ s.t. } 
  \CGSi,\assign[\var\mapsto\strat],\fplay\modelsSL \phi\\
 \CGSi,\assign,\fplay\modelsSL(a,\var)\phi & \text{ if } &
 \CGSi,\assign[a\mapsto\assign(\var)],\fplay\modelsSL \phi\vspace{5pt}\\
\end{array}\\
\mbox{and, writing $\iplay=\out(\assign,\fplay)$:}\\[5pt]

\begin{array}{p{72pt}cl}
  $\CGSi,\assign,\fplay\modelsSL\X\phi$ & \text{ if } &
  \CGSi,\assign,\iplay_{\leq |\fplay|}\modelsSL\phi\\
$\CGSi,\assign,\fplay\modelsSL\phi\until\phi'$ & \text{ if } & \exists\, i\geq 0
   \mbox{ s.t. }\CGSi,\assign,\iplay_{\leq |\fplay|+i-1}\modelsSL \phi'\\ 
   & & \text{ and } \forall\, j \text{ s.t. } 0\leq j <i,\;\\
& & \CGSi,\assign,\iplay_{\leq |\fplay|+j-1}\modelsSL \phi
\end{array}
\end{array}
\]
\end{definition}
To explain the temporal operators, we remind the reader that positions begin at $0$ and thus
$\pi_{n}$ is the $(n+1)$-st position of $\pi$.
Clearly, the satisfaction of a sentence is independent of the
assignment. For an \SLi sentence $\phi$ we thus let
$\CGSi,\fplay\models\phi$ if $\CGSi,\assign,\fplay\models\phi$ for some
assignment $\assign$, and we write $\CGSi \models \phi$ if $\CGSi, \pos_\init \models \phi$.

\subsection{Model checking and hierarchical instances}


\head{Model Checking} 
We now introduce the main decision problem of this paper, \ie, 
the model-checking problem for \SLi.  An \emph{\SLi-instance} is a formula/model pair $(\Phi,\CGSi)$ where $\Phi \in \SLi$ and $\CGSi$ is a \CGSi.
The \emph{model-checking problem} for \SLi is the decision problem that,
given an \SLi-instance $(\Phi,\CGSi)$, returns `yes' if
$\CGSi\models\Phi$, and `no' otherwise.

It is well known that
deciding the existence of winning  strategies in
multi-player games with
imperfect information is undecidable for reachability
objectives~\cite{peterson2001lower}. Since this problem is easily
reduced to the  model-checking problem for \SLi, we get the following result:

\begin{theorem}
  \label{theo-undecidable-SLi}
  The model-checking problem for \SLi is undecidable.
\end{theorem}

\head{Hierarchical instances}
We isolate a sub-problem obtained by restricting attention to \emph{hierarchical instances}. 
Intuitively, an \SLi-instance $(\Phi,\CGSi)$ is hierarchical if, as one goes down a path in the syntactic tree of $\Phi$, 
the observations tied to quantifications become finer. 
%
%

\begin{definition}[Hierarchical instances]
  \label{def-hierarchical-formula}
An \SLi-instance $(\Phi,\CGSi)$ is \emph{hierarchical} if
for all
  subformulas $\phi_{1},\phi_{2}$ of $\Phi$ of the form
  $\phi_{2}=\Estrato{\obs_{2}}\phi'_{2}$ and
  $\phi_{1}=\Estrato[y]{\obs_{1}}\phi'_{1}$ where  
  $\phi_{1}$  
  is a subformula of $\phi'_{2}$, it holds that 
  $\obsint(\obs_{1})\subseteq \obsint(\obs_{2})$.
\end{definition}

If $\obsint(\obs_{1})\subseteq \obsint(\obs_{2})$ we say that $\obs_{1}$ is \emph{finer} than $\obs_{2}$ in $\CGSi$, and that $\obs_{2}$ is \emph{coarser} than $\obs_{1}$ in $\CGSi$.  Intuitively, this means that a player with observation
$\obs_{1}$ observes game $\CGSi$ no worse than, \ie, is not less informed, 
a player with observation $\obs_{2}$.

\begin{example}[Security levels]
  We illustrate hierarchical instances in a ``security levels''
  scenario, where higher levels have access to more data (\ie, can
  observe more).  Assume that the \CGSi $\CGSi$ has
  $\obsint(\obs_{3})\subseteq \obsint(\obs_{2}) \subseteq
  \obsint(\obs_{1})$ (level $3$ is the highest security
  clearance, level $1$  is the lowest). Let
  $\phi = (a_1,x_1) (a_2,x_2) (a_3,x_3) \G p$. The \SLi formula
  $\Phi \egdef \Estrato[x_1]{\obs_1} \Astrato[x_2]{\obs_2}
  \Estrato[x_3]{\obs_3} \phi$ and $\CGSi$ together form a hierarchical instance. It expresses that  player $a_1$ (with lowest
  clearance) can collude with player $a_3$ (with highest
  clearance) to ensure a safety property $p$, even in the presence of
  an adversary $a_2$ (with intermediate clearance), as long as the
  strategy used by $a_3$ can depend on the strategy used by
  $a_2$.

  On the other hand,  formula
  $\Estrato[x_1]{\obs_1} \Estrato[x_3]{\obs_3} \Astrato[x_2]{\obs_2}
  \phi$, which is similar to $\Phi$ except that the strategy used by $a_3$ cannot depend
  on the adversarial strategy used by $a_2$, does not form a
  hierarchical instance  with $\CGSi$.
\end{example}


Here is the main contribution of this paper:

\begin{theorem}
\label{theo-SLi}
The model-checking problem for \SLi restricted to the class of hierarchical instances is decidable.
\end{theorem}

This is proved in Section~\ref{sec-modelcheckingSL} by reducing it to the model-checking 
problem of the hierarchical fragment of a logic called \QCTLsi, which we introduce, and 
prove decidable, in Section~\ref{sec-QCTL-imp-inf}.
We now give an important corollary of Theorem~\ref{theo-SLi}. 

A Nash equilibrium in a game is a tuple of strategies such that no
player has the incentive to deviate. Assuming that  $\Agf = \{a_i : i \in
[n]\}$ and goals are written
in \SLi, say $goal_i$ for $i \in [n]$,  the following formula 
of \SLi expresses the
existence of a Nash equilibrium:
\begin{align*}
 \Phi_\textsc{NE} \egdef & \Estrato[x_1]{\obs_1} \dots  \Estrato[x_n]{\obs_{n}}  (a_1,x_1) \dots (a_n,x_n) \Psi_\textsc{NE}
\end{align*}
where $\Psi_\textsc{NE} \egdef \bigwedge_{i \in [n]} \left[
 ( \Estrato[y_i]{\obs_i} (a_i,y_i) goal_i ) \to goal_i \right]$.

A \CGSi $\CGSi$ is said to \emph{yield hierarchical
  observation}~\cite{DBLP:conf/atva/BerwangerMB15} if
the ``finer-than'' relation is a total ordering, \ie, if for all
$\obs,\obs' \in \Obsf$, either $\obsint(\obs) \subseteq \obsint(\obs')$
or $\obsint(\obs') \subseteq \obsint(\obs)$.

Note that even if $\CGSi$ yields hierarchical information, the instance $(\Phi_\textsc{NE},\CGSi)$  is \emph{not}
 hierarchical (unless $\obsint(\obs_i) = \obsint(\obs_j)$ for all
$i,j \in [n]$). \emph{Nonetheless}, we can decide if a game that
yields hierarchical observation has a Nash equilibrium:

\begin{corollary}
The following problem is decidable: given a \CGSi that yields hierarchical observation, whether $\CGSi \models \Phi_\textsc{NE}$.
\end{corollary}
\begin{proof}
  The main idea is to use the fact that in a one-player game of
  partial-observation (such a game occurs when all but one player have
  fixed their strategies, as in the definition of Nash equilibrium),
  the player has a strategy enforcing some goal iff the player has a
  uniform strategy enforcing that goal. Here are the details.  Let
  $\CGSi=(\Act,\setpos,\trans,\val,\pos_\init,\obsint)$ be a \CGSi
  that yields hierarchical observation. Suppose the observation set is
  $\Obsf$.  To decide if $\CGSi \models \Phi_{\textsc{NE}}$ first
  build a new \CGSi
  $\CGSi'=(\Act,\setpos,\trans,\val,\pos_\init,\obsint')$ over
  observations $\Obsf' \egdef \Obsf \cup \{\obs_{p}\}$ such that
  $\obsint'(\obs) = \obsint(\obs)$ and $\obsint'(\obs_{p}) = \{(v,v) : v \in
  V\}$, and consider the sentence
\begin{align*} 
\Phi' \egdef & \Estrato[x_1]{\obs_1} \dots  \Estrato[x_n]{\obs_{n}}  (a_1,x_1) \dots (a_n,x_n) \Psi'
\end{align*}
where 
$\Psi' \egdef \bigwedge_{i \in [n]} \left[
 ( \Estrato[y_i]{\obs_{p}} (a_i,y_i) goal_i ) \to goal_i \right]$.

Then $(\Phi',\CGSi')$ is a hierarchical instance, and by Theorem~\ref{theo-SLi} we can decide $\CGSi' \models \Phi'$. We claim that 
$\CGSi' \models \Phi'$ iff $\CGSi \models \Phi_{\textsc{NE}}$. To see this, it is enough to establish that: 
\[ \CGSi',\assign,\pos_{\init} \models \Estrato[y_i]{\obs_{p}} (a_i,y_i) goal_i \equivaut \Estrato[y_i]{\obs_i} (a_i,y_i) goal_i,
\]
for every $i \in [n]$ and every assignment $\assign$ such that $\assign(x_i) = \assign(a_i)$ is an $\obs_i$-uniform strategy.

To this end, fix $i$ and $\assign$. The right-to-left implication is
immediate (since $\obs_{p}$ is finer than $\obs_i$).  For the converse, let
$\sigma$ be a $p$-uniform strategy (\ie, perfect-information) such
that $\CGSi',\assign[y_i \mapsto \sigma,a_i \mapsto \sigma],\pos_{\init} \models
goal_i$.
Let $\pi \egdef \out(\assign[y_i \mapsto \sigma,a_i \mapsto
\sigma],\pos_\init)$. Construct an $\obs_i$-uniform strategy $\sigma'$
that agrees with $\sigma$ on prefixes of $\pi$. This can be done as
follows: if $\rho \obseq[\obs_{i}] \pi_{\leq j}$ for some $j$ then define
$\sigma'(\rho) = \sigma(\pi_{\leq j})$ (note that this is well-defined
since if there is some such $j$ then it is unique), and otherwise
define $\sigma'(\rho) = a$ for some fixed action $a \in \Act$.
\end{proof}

\subsection{Comparison with other logics}
\label{subsec:SLi-comparison}

The main difference between \SL and \ATL-like strategic logics is that
in the latter a strategy is always bound to some player, while in the
former bindings and quantifications are separated. This separation
adds expressive power, e.g., one can bind the same strategy to
different players. Extending \ATL with imperfect-information is done
by giving each player an indistinguishability relation that its
strategies must respect~\cite{BJ14}.  Our extension of \SL by
imperfect information, instead, assigns each strategy $x$ an
indistinguishability relation $o$ when it is quantified
$\Estrato{\obs}$. Thus $\Estrato{\obs}\phi$ means ``there exists a
strategy with observation $\obs$ such that $\phi$ holds''. Associating
observations in this way, \ie, to strategies rather than players has
two consequences. First, it is a clean generalisation of \SL in the
perfect information setting~\cite{DBLP:journals/tocl/MogaveroMPV14}.
Define the \emph{perfect-information fragment of \SLi} to be the logic
\SLi  assuming that $\Obsf = \{\obs\}$ and 
$\obsint(\obs) = \{(v,v) : v \in \CGSi\}$ for every \CGSi $\CGSi$; also let us
assimilate such structures with classic perfect-information concurrent
game structures (\CGS), which are the models of \SL. Finally, let
$\trSLSLi:\SL\to\SLi$ be the trivial translation that annotates
each strategy quantifier $\Estrato{}$ with observation $\obs$. The next proposition says
that the perfect-information fragment of \SLi is a notational variant
of \SL. 


\begin{proposition}\label{prop-perfectinformation}
  For every \SL sentence $\varphi$ and every \CGS $\CGS$, it holds
  that $\CGS
  \models \varphi$ iff $\CGSi \models \trSLSLi(\varphi)$.
\end{proposition}



Second, \SLi subsumes imperfect-information extensions of \ATLs that
associate observations to players.

We recall that an \ATLsi formula\footnote{See~\cite{BJ14} for the definition of \ATLsi, where subscript i
refers to ``imperfect information'' and subscript R to ``perfect
recall''. Also, we consider the so-called \emph{objective
    semantics} for \ATLsi.} $\EstratATL \psi$  reads as ``there are
strategies for players in $\coal$ such that $\psi$ holds whatever
players in $\Agf\setminus\coal$ do''.
Consider the translation $\trATLSL:\ATLsi\to\SLi$ that replaces each subformula of
the form $\EstratATL \psi$, where
$\coal=\{\ag_{1},\ldots,\ag_{k}\}\subset\Agf$ is a coalition of
players and $\Agf\setminus\coal = \{\ag_{k+1},\ldots,\ag_{n}\}$, with formula $\Estrato[\var_{1}]{\obs_{1}}\ldots
\Estrato[\var_{k}]{\obs_{k}} \Astrato[\var_{k+1}]{\obs_{p}}\ldots \Astrato[\var_{n}]{\obs_{p}}
\bind[\ag_{1}]{\var_{1}}\ldots \bind[\ag_{n}]{\var_{n}} \psi'$, where
$\psi'=\trATLSL(\psi)$. Also, for every \CGSi as considered in the
semantics of \ATLi, \ie, where each agent is assigned an equivalence
relation on positions (let us refer to such structures as \CGSiATL),
define the \CGSi  $\CGSi'$  by interpreting each $\obs_{i}$ as the equivalence relation
for agent~$\ag_i$ in $\CGSi$, and interpreting $\obs_{p}$ as the identity
relation.

\begin{proposition}
  \label{prop-subsume-ATL}
  For every \ATLsi formula $\phi$ and \CGSiATL $\CGSi$,
  it holds that $\CGSi\models \phi$ iff $\CGSi'\models \trATLSL(\phi)$.
\end{proposition}


Third, \SLi also subsumes the imperfect-information extension of \ATLs
with strategy context
(see~\cite{DBLP:journals/corr/LaroussinieMS15} for the definition of
\ATLssc with partial observation, which we refer to as \ATLssci). 
The only difference between \ATLssci and \ATLsi is the following: in
\ATLsi, when a subformula of the form $\EstratATL \psi$ is met, we
quantify existentially on strategies for players in $\coal$, and
then we consider all possible outcomes obtained by letting other
players behave however they want. Therefore, if any player in
$\Agf\setminus\coal$ had previously been assigned a strategy, it is
forgotten. In \ATLssci on the other hand, these strategies are stored
in a \emph{strategy context}, which is a \emph{partial} assignment $\assign$, defined for the
subset of players currently bound to a strategy.
A strategy context allows one to quantify universally only on strategies
 of players who are not in $\coal$ and who are not already
bound to a strategy. It is then easy to define a translation
$\trATLscSL:\ATLssci\to\SLi$  by adapting translation $\trATLSL$
from Proposition~\ref{prop-subsume-ATL}, with the
 strategy context as parameter. For an \CGSiATL $\CGSi$, the \CGSi $\CGSi'$ is defined as for Proposition~\ref{prop-subsume-ATL}.

\begin{proposition}
  \label{prop-subsume-ATLsc}
  For every \ATLssci formula $\phi$ and \CGSiATL $\CGSi$,
    it holds that $\CGSi\models \phi$ iff $\CGSi'\models \trATLscSL(\phi)$.
\end{proposition}

Fourth, there is a natural and simple translation of instances of the
model-checking problem of \CL\cite{DBLP:conf/csl/FinkbeinerS10} into
the hierarchical instances of \SLi.  Moreover, the image of this
translation consists of instances of \SLi with a very restricted form,
\ie, atoms mentioned in the \SLi-formula are \emph{observable} for all
observations of the \CGSi, \ie,
players know the truth value of all atoms in all positions, for
any observation they are assigned. 

\begin{proposition} \label{prop:CLtoSLi}
There is an effective translation that, given a \CL-instance $(\CKS,\phi)$ produces a hierarchical \SLi-instance $(\CGSi,\Phi)$ such that
\begin{enumerate}
 \item $\CKS \models \phi$ iff $\CGSi \models \Phi$,
 \item For all atoms $p$  in $\Phi$,  observations $\obs \in
   \Obsf$ and positions
   $\pos,\pos'\in\CGSi$, if $\pos \obseq \pos'$ then $p \in \val(\pos) \iff p \in \val(\pos')$. 
\end{enumerate}
\end{proposition}

To do this, one first translates \CL into (hierarchical) \QCTLsi, the
latter is defined in the next section. This step is a simple
reflection of the semantics of \CL in that of \QCTLsi. Then one
translates \QCTLsi into \SLi by a simple adaptation of the translation
of \QCTLs into
\ATLssc~\cite{DBLP:journals/iandc/LaroussinieM15}. 

\section{\QCTLs with imperfect information}
\label{sec-QCTL-imp-inf}

In this section we introduce an imperfect-information extension of
\QCTLs~\cite{Sis83,emerson1984deciding,Kup95,KMTV00,french2001decidability,
DBLP:journals/corr/LaroussinieM14}. In order to introduce imperfect
information, instead of considering equivalence relations between
states as in concurrent game structures, we
will enrich Kripke structures by giving internal structure to their states,
\ie, we see states as $n$-tuples of local states. This way of
modelling imperfect information is inspired from
  Reif's multi-player
  game structures~\cite{peterson2001lower} and distributed
  systems~\cite{halpern1989complexity}, and we find it very suitable to application of
  automata techniques, as discussed in
  Section~\ref{sec-discussion-QCTL}.

  The syntax of \QCTLsi is 
similar to that of \QCTLs, except that we annotate second-order quantifiers by
subsets $\cobs \subseteq [n]$. The idea is that quantifiers annotated by $\cobs$ 
can only ``observe'' the local states indexed by $i \in \cobs$. 
We define the tree-semantics of \QCTLsi: this means that we
interpret formulas on trees that are the unfoldings of Kripke structures (this
will capture the fact that players in \SLi have synchronous perfect recall).

We then define the syntactic class of \emph{hierarchical formulas} and prove,
using an automata-theoretic approach, that  model checking this
class of formulas is decidable.

\subsection{\QCTLsi Syntax}
\label{sec-syntax-QCTLi}

The syntax of \QCTLsi is very similar to that of \QCTLs: the only difference is
that we annotate quantifiers by a set of indices that defines the
``observation'' of that quantifier.

\begin{definition}[\QCTLsi Syntax]
  \label{def-syntax-QCTLsi}
  Fix 
  $n \in \setn$.
  The syntax of \QCTLsi is defined by the following grammar:
  \begin{align*}
  \phi\egdef &\; p \mid \neg \phi \mid \phi\ou \phi \mid \E \psi \mid
  \existsp[p]{\cobs} \phi\\
    \psi\egdef &\; \phi \mid \neg \psi \mid \psi\ou \psi \mid \X \psi \mid
  \psi \until \psi
\end{align*}
where $p\in\APf$ 
and $\cobs\subseteq [n]$.
\end{definition}

Formulas of type $\phi$ are called \emph{state formulas}, those of type $\psi$
are called \emph{path formulas}, and \QCTLsi consists of all the state formulas
defined by the grammar.   We use standard abbreviation 
$\A\psi \egdef \neg\E\neg\psi$.
We also use $\exists p.\,\phi$ as a shorthand for $\existsp[p]{[n]}
\phi$, 
and we let $\forall p.\, \phi\egdef \neg \exists p.\, \neg \phi$.

Given a \QCTLsi formula $\phi$, we define the set of \emph{quantified
  propositions} $\APq(\phi)\subseteq\APf$ as the set of atomic propositions $p$ such that
$\phi$ has a subformula of the form $\existsp[p]{\cobs}\phi$. We also
define the set of \emph{free propositions} $\APfree(\phi)\subseteq\APf$ as the set
of atomic propositions that appear out of the scope of any quantifier
of the form $\existsp[p]{\cobs}$ Observe that 
$\APq(\phi)\inter\APfree(\phi)$ may not be empty in general, \ie, a proposition may appear both free and quantified in (different
places of) a formula. 

\subsection{\QCTLsi tree-semantics}
\label{sec-QCTLsi-semantics}
We define the semantics on structures whose states are tuples of local states.

\head{Local states}
Let $\{\setlstates_{i}\}_{i\in [n]}$ denote $n \in \setn$ disjoint finite sets of
\emph{local states}. For  $I\subseteq [n]$, we let $\Dirtreei\egdef\prod_{i\in
I}\setlstates_{i}$ if $I\neq\emptyset$, and
$\Dirtreei[\emptyset]\egdef\{\blank\}$ where $\blank$ is a special symbol.

\head{Concrete observations}
A set $\cobs \subseteq [n]$ is called a \emph{concrete observation} (to distinguish
it from observations $\obs$ in the definitions of \SLi).

\head{Compound Kripke structures} 
These are like Kripke structures
except that the states are elements of $\Dirtreei[{[n]}]$.
A \emph{compound Kripke structure}, or \CKS, over $\APf$, is a tuple 
$\CKS=(\setstates,\relation,\state_\init,\lab)$ where
$\setstates\subseteq \Dirtreei[{[n]}]$ is a set of \emph{states}, 
$\relation\subseteq\setstates\times\setstates$ is a
left-total\footnote{\ie, for all $\state\in\setstates$, there exists $\state'$
such that $(\state,\state')\in\relation$.} \emph{transition relation}, 
$\state_\init \in \setstates$ is an \emph{initial state}, and
$\lab:\setstates\to 2^{\APf}$ is a \emph{\labeling function}.

A \emph{path} in $\CKS$  is an infinite sequence of states
$\spath=\state_{0}\state_{1}\ldots$ such that
 for all $i\in\setn$,
$(\state_{i},\state_{i+1})\in \relation$. For 
$\state\in\setstates$, we let $\Paths(\state)$ be the set of all
paths that start in $\state$.  
A \emph{finite path} is a finite non-empty prefix of a path.
We may write $\state\in\CKS$ for $\state\in\setstates$.
Since we will interpret \QCTLsi on unfoldings of \CKS, we now define
infinite trees.

\head{Trees}
In many works, trees are defined as prefix-closed sets of words 
with the empty word $\epsilon$ as root. Here trees represent unfoldings of 
Kripke structures, and we find it more convenient to see a node $u$ as a sequence of
states and the root as the initial state. 
Let $\Dirtree$ be a finite set (typically a set of states). 
An \emph{$\Dirtree$-tree} $\tree$ 
 is a
nonempty set of words $\tree\subseteq \Dirtree^+$ such that:
\begin{itemize}
  \item\label{p-root} there exists $\racine\in\Dirtree$,  called the
    \emph{root} of $\tree$, such that each
    $\noeud\in\tree$ starts with $\racine$ ($\racine\pref\noeud$);
  \item if $\noeud\cdot\dir\in\tree$ and $\noeud\neq\epsilon$, then
    $\noeud\in\tree$, and
  \item if $\noeud\in\tree$ then there exists $\dir\in\Dirtree$ such that $\noeud\cdot\dir\in\tree$.
\end{itemize}

The elements of a tree $\tree$ are called \emph{nodes}.  
  If 
 $\noeud\cdot\dir \in \tree$, we say that $\noeud\cdot\dir$ is a \emph{child} of
 $\noeud$.
 The \emph{depth} of a node $\noeud$ is $|\noeud|$.
An $\Dirtree$-tree $\tree$ is \emph{complete} if $\noeud \in \tree, \dir \in \Dirtree$ implies $\noeud \cdot \dir \in \tree$.
A \emph{\tpath} in $\tree$ is an infinite sequence of nodes $\tpath=\noeud_0\noeud_1\ldots$
such that for all $i\in\setn$, $\noeud_{i+1}$ is a child of
$\noeud_i$,
and $\tPaths(\noeud)$ is the set of \tpaths
 that start in node $\noeud$. 
An \emph{$\APf$-\labeled $\Dirtree$-tree}, or
\emph{$(\APf,\Dirtree)$-tree} for short, is a pair
$\ltree=(\tree,\lab)$, where $\tree$ is an $\Dirtree$-tree called the
\emph{domain} of $\ltree$ and
$\lab:\tree \rightarrow 2^{\APf}$ is a \emph{\labeling}.
For a labelled tree $\ltree=(\tree,\lab)$ and an atomic
proposition $p\in\APf$, we define the \emph{$p$-projection of $\ltree$}
as the labelled tree
$\proj{\ltree}\;\egdef(\tree,\proj{\lab})$, where for each
$\noeud\in\tree$, $\proj{\lab}\!(\noeud)\egdef \lab(\noeud)\setminus
\{p\}$. For a set of trees $\lang$, we let $\proj{\lang}\;\egdef\{\proj{\ltree}\;\mid\ltree\in\lang\}$.
Finally, two labelled trees $\ltree=(\tree,\lab)$ and
  $\ltree'=(\tree',\lab')$ are \emph{equivalent modulo $p$},
  written $\ltree\Pequiv\ltree'$, if
  $\proj{\ltree}=\proj{\ltree'}$ (in particular, $\tree=\tree'$).
  
%

%
\head{Narrowing}
Let $\Dirtree$ and $\Dirtreea$ be two finite sets, and let $(\dir,\dira)\in\Dirtree\times\Dirtreea$.
 The
 \emph{$\Dirtree$-narrowing} of $(\dir,\dira)$ is
${\projI[\Dirtree]{(\dir,\dira)}}\egdef \dir$.
%
This definition extends naturally to words over
$\Dirtree\times\Dirtreea$ (pointwise), and thus to
$\Dirtree\times\Dirtreea$-trees.

 For $J\subseteq I\subseteq [n]$ and $\dirz=(\lstate_{i})_{i\in
   I}\in\Dirtreei[I]$,
 we also define 
 ${\projI[{J}]{\dirz}}\egdef
 \projI[{\Dirtreei[J]}]{\dirz}$,
 where $\dirz$ is seen as a pair $\dirz=(\dirz_{1},\dirz_{2})\in
 \Dirtreei[J]\times \Dirtreei[{I\setminus J}]$. This is well defined
 because having taken sets $\setlstates_{i}$ to be disjoint, the
 ordering of local states in $\dirz$ is indifferent.
 We also extend this definition to words and trees.
 
Observe that when narrowing a tree,
nodes with same narrowing are merged. In particular, for every
$\Dirtreei[I]$-tree $\tree$, $\projI[\emptyset]{\tree}$ is the only
$\Dirtreei[\emptyset]$-tree, $\blank^{\omega}$.


\head{Quantification and uniformity} 
In \QCTLsi the intuitive meaning of $\existsp[p]{\cobs} \phi$ in a tree $t$ is that there is some equivalent tree $t'$ modulo $p$ such that
$t'$ is $\cobs$-uniform in $p$ and satisfies $\phi$.  Intuitively, a tree is
$\cobs$-uniform in $p$ if it is uniformly labelled by $p$, \ie, 
if every two nodes that are indistinguishable
when projected onto the local states indexed by $\cobs \subseteq [n]$ agree on
their labelling of $p$.

\begin{definition}[$\cobs$-indistinguishability and $\cobs$-uniformity
  in $p$]
  \label{def-uniformity}
Fix $\cobs \subseteq [n]$ and $I \subseteq [n]$.

\begin{itemize}
 \item Two tuples $\dir,\dir'\in\Dirtreei[I]$ are \emph{$\cobs$-indistinguishable},
written $\dir\oequiv\dir'$, if 
$\projI[{I\cap\,\cobs}]{\dir}=\projI[{I\cap\,\cobs}]{\dir'}$.

\item Two words
   $\noeud=\noeud_{0}\ldots\noeud_{i}$ and
   $\noeud'=\noeud'_{0}\ldots\noeud'_{j}$ over alphabet $\Dirtreei[I]$ are
   \emph{$\cobs$-indistinguishable}, written $\noeud\oequivt\noeud'$, if
   $i=j$ and for all $k\in \{0,\ldots,i\}$ we have
   $\noeud_{k}\oequiv\noeud'_{k}$.

\item      A tree $\ltree$ is \emph{$\cobs$-uniform in $p$} if for every pair of nodes
 $\noeud,\noeud'\in\tree$ such that $\noeud\oequivt\noeud'$, we have
 $p\in\lab(\noeud)$ iff $p\in\lab(\noeud')$.

\end{itemize}
\end{definition}

Finally, we inductively define the satisfaction relation $\modelst$ for the semantics on trees, where   $\ltree=(\tree,\lab)$ is
a $2^{\APf}$-\labeled $\Dirtreei$-tree, 
$\noeud$ is a node and $\tpath$ is a path in $\tree$:
\begin{align*}
  \ltree,\noeud\modelst & 	\,p 			& \mbox{ if } &\quad p\in\lab(\noeud)\\
  \ltree,\noeud\modelst & 	\,\neg \phi		& \mbox{ if } & \quad\ltree,\noeud\not\modelst \phi\\
  \ltree,\noeud\modelst & 	\,\phi \ou \phi'		& \mbox{ if } &\quad \ltree,\noeud \modelst \phi \mbox{ or    }\ltree,\noeud\modelst \phi' \\
  \ltree,\noeud\modelst & 	\,\E\psi			& \mbox{ if } &\quad \exists\,\tpath\in\tPaths(\noeud) \mbox{      s.t. }\ltree,\tpath\modelst \psi \\
  \ltree,\noeud\modelst & \,\existsp{\cobs} \phi & \mbox{ if }
  & \quad \exists\,\ltree'\Pequiv[p]\ltree \mbox{ s.t.  }\ltree'\mbox{ is $\cobs$-uniform in $p$ and }\\ & &  & \quad\ltree',\noeud\modelst\phi.\\
\ltree,\tpath\modelst &		\,\phi 			& \mbox{ if } &\quad \ltree,\tpath_{0}\modelst\phi \\ 
\ltree,\tpath\modelst &		\,\neg \psi 		& \mbox{ if }
& \quad \ltree,\tpath\not\modelst \psi \\ 
\ltree,\tpath\modelst & \,\psi \ou \psi'			& \mbox{ if } & \quad\ltree,\tpath \modelst \psi \mbox{ or }\ltree,\tpath\modelst \psi' \\ 
\ltree,\tpath\modelst & \,\X\psi 				& \mbox{ if } & \quad\ltree,\tpath_{\geq 1}\modelst \psi \\ 
\ltree,\tpath\modelst & \,\psi\until\psi' 		& \mbox{ if } & \quad\exists\, i\geq 0 \mbox{ s.t.    }\ltree,\tpath_{\geq i}\modelst\psi' \text{ and} \\ 
&&& \quad\forall j \text{ s.t. }0\leq j <i,\; \ltree,\tpath_{\geq j}\modelst\psi
\end{align*}


We write $\ltree\modelst\phi$ for $\ltree,\racine\modelst\phi$,
where $\racine$ is the root of $\ltree$.

\begin{example}
  \label{example-QCTLi}
  Consider the following \CTL formula:
\[\ligne{p}\egdef  \A \F p \wedge \A \G (p
\rightarrow \A\X \A\G \neg p).\]

This formula  holds in a labelled tree if and only if each path contains
exactly one node labelled with $p$. Now, consider the following
\QCTLsi formula:
\[\ligneb{p}\egdef\existsp{\emptyset}\ligne{p}.\]
For a blind quantifier, two nodes of a tree are
indistinguishable if and only if they have same depth. Therefore,
this formula holds on a tree iff the $p$'s label all and only the nodes at 
some fixed depth. This formula can thus be used to capture the equal
level predicate on trees. Actually, just as \QCTLs captures \MSO, 
one can prove that \QCTLsi with tree semantics subsumes \MSO with
equal level~\cite{elgot-rabin66,lauchli1987monadic,thomas-msoeqlevel}.
In Theorem~\ref{theo-undecidable} we make use of a similar observation to prove that
model-checking \QCTLsi is undecidable.
\end{example}

\head{Model-checking problem for \QCTLsi under tree semantics}
For the model-checking problem, we interpret \QCTLsi on unfoldings of \CKSs.

\head{Tree unfoldings $\unfold{\state}$} 
Let $\CKS=(\setstates,\relation,\state_\init,\lab)$ be a compound Kripke structure over $\APf$, and let $\state\in\setstates$. 
The \emph{tree-unfolding of $\CKS$ from $\state$} is the $(\APf,\setstates)$-tree $\unfold{\state}\egdef (\tree,\lab')$, where
    $\tree$ is the set
    of all finite  paths that start in $\state$, and
    for every $\noeud\in\tree$,
    $\lab'(\noeud)\egdef \lab(\last(\noeud))$.
Given a \CKS $\KS$, a state $\state\in\CKS$  and a
\QCTLsi formula $\phi$, we write $\KS,\state\modelst\phi$ if
$\unfold[\KS]{\state}\modelst\phi$. Write $\KS \modelst \phi$ if $\unfold[\KS]{\state_\init} \models \phi$.

The \emph{model-checking problem for \QCTLsi under tree-semantics} is the 
following decision problem: given an instance $(\phi,\CKS)$ where $\CKS$ is a CKS, 
and $\phi$ is a \QCTLsi formula, return `Yes' if $\CKS \modelst \phi$ and `No' otherwise.



\subsection{Discussion of the definition of \QCTLsi}
 \label{sec-discussion-QCTL}

 We now motivate in detail some aspects of \QCTLsi.
 
\head{Modelling of imperfect information} 
We model imperfect information
 by means of local states (rather than equivalence relations) 
 since this greatly facilitates the use of
 automata techniques. More precisely,
 in our decision procedure of Section~\ref{sec-decidable},
 we make extensive use of an operation on tree automata called
 \emph{narrowing}, which was introduced 
 in~\cite{kupferman1999church} to deal with imperfect-information
 in the context of distributed synthesis for temporal
 specifications. Given an automaton $\auto$ that works on $\Dirtree\times
 \Dirtreea$-trees, where $\Dirtree$ and $\Dirtreea$ are two finite sets, and assuming that we
 want to model an operation performed on trees while observing only
 the $\Dirtree$ component of each node,
this narrowing operation allows one to build from $\auto$ an automaton
$\auto'$ that works
on $\Dirtree$-trees, such that $\auto'$ accepts an $\Dirtree$-tree if
and only if $\auto$ accepts its widening to $\Dirtree\times\Dirtreea$
(see Section~\ref{sec-decidable} for details). One can then make this
automaton $\auto'$ perform the desired
operation, which will by necessity be performed uniformly with regards
to the partial observation, since the $\Dirtreea$ component is absent from the
input trees.

With our definition of compound Kripke structures, their
unfoldings are trees over the  Cartesian product
$\Dirtreei[{[n]}]$. To model a quantification  $\exists^{\cobs}p$ with observation
$\cobs\subseteq [n]$, we can thus use the narrowing operation to
forget about components $\setlstates_{i}$, for $i\in [n]\setminus\cobs$.
We then use the classic projection of nondeterministic tree automata to perform
existential quantification on atomic proposition $p$. Since the choice
of the $p$-labelling is made directly on $\Dirtreei[\cobs]$-trees, it
is necessarily $\cobs$-uniform. 

\head{Choice of the tree semantics}
 \QCTLs is obtained by adding to \CTLs second-order
 quantification on atomic propositions. Several semantics have been
 considered. The two most studied ones are the \emph{structure
   semantics}, in which  formulas are evaluated directly on Kripke
 structures, and the \emph{tree semantics}, in which Kripke structures
 are first unfolded into infinite trees. Tree semantics thus allows
 quantifiers to
 choose the value of a quantified atomic proposition in each
 \emph{finite path} of the model, while in structure semantics the
 choice is only made in each state. 
When \QCTLs is used to express existence of strategies, existential
quantification on atomic propositions labels the structure with
strategic choices; in this kind of application, structure semantics
reflects so-called \emph{positional} or \emph{memoryless} strategies,
while tree semantics captures \emph{perfect-recall} or
\emph{memoryfull} strategies. Since in this work we are interested in
perfect-recall strategies, we only consider the tree semantics.

\subsection{Model checking \QCTLsi}
\label{sec-modelchecking}

We now prove that the model-checking problem for \QCTLsi
under tree semantics is undecidable. This comes as 
no surprise since, as we will show, \QCTLsi can  express the existence of
 winning strategies in imperfect-information games.

 \newcounter{theo-undecidable}
\setcounter{theo-undecidable}{\value{theorem}}

\begin{theorem}
    \label{theo-undecidable}
The model-checking problem for \QCTLsi under tree-semantics is undecidable.
\end{theorem}
\begin{proof}
Let \MSOeql denote the extension of the logic \MSO  by a
  binary predicate symbol $\eql$. Formulas of \MSOeql are interpreted on trees, 
  and 
  the semantics of $\eql(x,y)$ is that $x$ and $y$ have the same depth in
  the tree. There is a translation of \MSO-formulas to \QCTLs-formulas that preserves satisfaction~\cite{DBLP:journals/corr/LaroussinieM14}. 
  This translation can be extended to map \MSOeql-formulas to \QCTLsi-formulas using the formula $\ligneb{\cdot}$
  from Example~\ref{example-QCTLi} to help capture the equal-length predicate.
 Our result follows since the \MSOeql-theory of the binary tree is undecidable~\cite{lauchli1987monadic}.
\end{proof}

\section{A decidable fragment of \QCTLsi: hierarchy on observations}
\label{sec-decidable}

The main result of this section is the identification of an important decidable fragment of \QCTLsi.

\begin{definition}[Hierarchical formulas]
  \label{def-hierarchical}
  A \QCTLsi formula $\phi$ is \emph{hierarchical} if for all
  subformulas $\phi_{1},\phi_{2}$ of the form
  $\phi_{1}=\existsp[p_{1}]{\cobs_{1}}\phi'_{1}$ and
  $\phi_{2}=\existsp[p_{2}]{\cobs_{2}}\phi'_{2}$ where  
  $\phi_{2}$
  is a subformula of $\phi'_{1}$, we have $\cobs_{1}\subseteq\cobs_{2}$.
\end{definition}

In other words, a formula is hierarchical if innermost propositional
quantifiers observe at least as much as  outermost ones.
We let \QCTLsih be the set of hierarchical \QCTLsi formulas.

\begin{theorem}
  \label{theo-decidable-QCTLi}
Model checking \QCTLsih under tree semantics is non-elementary decidable.
\end{theorem}

Since our decision procedure for the hierarchical fragment
of \QCTLsi is based on an automata-theoretic approach, we
recall some definitions and results for alternating tree automata.

\subsection{Alternating parity tree automata}
\label{sec-ATA}


  We briefly recall the notion of alternating (parity) tree automata.
  Because it is sufficient for our needs and simplifies definitions,
  we assume that all input trees are complete trees.
For a set $Z$, $\boolp(Z)$ 
is the set of
formulas built from the elements of $Z$ as atomic propositions using the connectives $\ou$ and
$\et$, 
and with $\top,\perp \in \boolp(Z)$.
An \emph{alternating tree automaton (\ATA) on $(\APf,\Dirtree)$-trees}
is a structure $\auto=(\tQ,\tdelta,\tq_{\init},\couleur)$
where 
$Q$ is a finite set of states, $\tq_{\init}\in \tQ$ is an initial
state, $\tdelta : \tQ\times 2^{\APf} \rightarrow \boolp(\Dirtree\times
\tQ)$ is a transition function, and $\couleur:\tQ\to \setn$ is a
colouring function.  To ease reading we shall write atoms in
$\boolp(\Dirtree\times\tQ)$ between brackets, such as
$[x,\tq]$.  
A \emph{nondeterministic tree automaton (\NTA) on
  $(\APf,\Dirtree)$-trees} is an \ATA
$\auto=(\tQ,\tdelta,\tq_{\init},\couleur)$ such that for every $\tq\in
\tQ$ and $a\in 2^{\APf}$, $\tdelta(\tq,a)$ is written in
disjunctive normal form and for every direction $\dir\in \Dirtree$ 
each disjunct contains exactly one element of $\{\dir\}\times Q$. 
Acceptance is defined as usual (see, \eg,
\cite{DBLP:journals/tcs/MullerS87}), and we let $\lang(\auto)$ be the
set of trees accepted by $\auto$.

We recall three classic results on tree automata. The first one is
that nondeterministic tree automata are closed under projection, and
was established by
Rabin to deal with second-order monadic quantification:
\begin{theorem}[Projection \cite{rabin1969decidability}]
  \label{theo-projection}
  Given an \NTA $\NTA$ and an atomic
  proposition $p\in\AP$, one can build an \NTA $\proj{\NTA}$ 
  such that
  $\lang(\proj{\NTA})=\proj{\lang(\NTA)}$. 
\end{theorem}

Because it will be important to understand the automata construction
for our decision procedure in Section~\ref{sec-decidable}, we briefly recall
that the projected automaton ${\proj{\NTA}}$ is simply automaton $\NTA$ 
with the only difference that when it reads the label of a node, it
can choose whether $p$ is there or not: if $\tdelta$ is the transition
function of $\NTA$, that of ${\proj{\NTA}}$ is
$\tdelta'(q,a)=\tdelta(q,a\union \{p\}) \ou
\tdelta(q,a\setminus\{p\})$, for any state $q$ and label $a\in 2^{\APf}$. Another way of seeing it is that
$\proj{\NTA}$ first guesses a $p$-labelling for the input tree, and
then simulates $\NTA$ on this modified input.
To prevent $\proj{\NTA}$ from guessing different labels for a same
node in different executions, it is crucial that $\NTA$ be nondeterministic, 
which is the reason why we need the next classic result: 
 the crucial simulation theorem, due to Muller and
Schupp.
\begin{theorem}[Simulation \cite{DBLP:journals/tcs/MullerS95}]
\label{theo-simulation}
Given an \ATA $\ATA$, one can build an \NTA $\NTA$ 
 such that $\lang(\NTA)=\lang(\ATA)$.
\end{theorem}

The third result was established by Kupferman and Vardi to deal with
imperfect information aspects in distributed synthesis. To state it we need
to define a widening operation on trees which simply expands the directions
in a tree.

\head{Widening~\cite{kupferman1999church}} 
Let $\Dirtree$ and $\Dirtreea$ be two finite sets, let $\ltree$ be
an $\Dirtree$-tree with root $\dir$, and let $\dira\in\Dirtreea$. The
\emph{$\Dirtreea$-widening} of $\ltree$ with root $(\dir,\dira)$ is
the $\Dirtree\times\Dirtreea$-tree
  \[{\liftI[\Dirtreea]{\dira}{\tree}}\egdef
  \{\noeud\in (\dir,\dira) \cdot(\Dirtree\times\Dirtreea)^{*}\mid
  {\projI[\Dirtree]{\noeud}}\in\tree \}.\]

For an
 $(\APf,\Dirtree)$-tree
$\ltree=(\tree,\lab)$, we let $\liftI[\Dirtreea]{\dira}{\ltree}\egdef (\liftI[\Dirtreea]{\dira}{\tree},\lab')$
where $\lab'(\noeud)\egdef \lab(\projI[\Dirtree]{\noeud})$. 

Similarly to the narrowing operation, we extend this definition to tuples of local
states by letting, for $J\subseteq I\subseteq
[n]$,  $\tree$ an
 $\Dirtreei[J]$-tree  and $\dirz'\in
 \Dirtreei[I\setminus J]$,
\[\liftI[I]{\dirz'}{\tree} \egdef\liftI[{\Dirtreei[{I\setminus J}]}]{\dirz'}{\tree},\]
and similarly for a labelled $\Dirtreei[J]$-tree $\ltree$,
\[\liftI[I]{\dirz'}{\ltree} \egdef\liftI[{\Dirtreei[{I\setminus J}]}]{\dirz'}{\ltree}.\]
Recall that because the sets of local states $\setlstates_{i}$ are
disjoint, the order of local states in a tuple does not matter and we
can identify $\Dirtreei[I]$ with
$\Dirtreei[J]\times\Dirtreei[I\setminus J]$.

The rough idea of the narrowing operation on \ATA is
that,  if one just observes
 $\Dirtree$,  uniform  $p$-labellings on 
$\Dirtree\times\Dirtreea$-trees  can be obtained by choosing the
labellings directly on $\Dirtree$-trees, and then lifting them to $\Dirtree\times\Dirtreea$.

\begin{theorem}[Narrowing \cite{kupferman1999church}]
  \label{theo-narrow}
  Given an \ATA $\ATA$ on $\Dirtree\times\Dirtreea$-trees 
  one can build an \ATA ${\narrow[\Dirtree]{\ATA}}$ on $\Dirtree$-trees
 such that 
 for all $\dira\in\Dirtreea$, $\ltree\in\lang(\narrow[\Dirtree]{\ATA})$ iff $\liftI[\Dirtree\times\Dirtreea]{\dira}{\ltree}\in\lang(\ATA)$.
\end{theorem}

\subsection{Translating \QCTLsih to \ATA}
In order to prove Theorem~\ref{theo-decidable-QCTLi} we need some more
notations and a technical lemma that contains the automata construction.

For every $\phi\in\QCTLsi$, we let \[\Iphi\egdef
\biginter_{\cobs\in\setobs}\cobs \subseteq [n],\] where $\setobs$ is the set of
concrete observations
  that occur in $\phi$, with the intersection over the empty set 
  defined as $[n]$.  We also let $\Dirtreei[\phi]\egdef
  \Dirtreei[\Iphi]$ (recall that for $I\subseteq [n]$,
  $\Dirtreei=\prod_{i\in I}\setlstates_{i}$).

Our construction, that transforms a \QCTLsih formula $\phi$ and a \CKS
$\CKS$ into an \ATA,
builds upon the classic construction 
  from~\cite{DBLP:journals/jacm/KupfermanVW00}, which builds hesitant
  \ATA for \CTLs formulas. Since our aim here is to establish
  decidability and that the hesitant
  condition is only used to improve complexity, we do not consider it. However
  we need to develop an original technique to implement
  quantifiers with imperfect information thanks to automata narrowing and projection.

  The classical approach to model checking via tree automata is to
  build an automaton that accepts all tree models of the input
  formula,  and check whether it accepts the unfolding of the
  model~\cite{DBLP:journals/jacm/KupfermanVW00}.
  We now explain how we adapt this approach.

  \head{Narrowing of non-uniform trees} Quantification on atomic
  propositions is classically performed by means of automata
  projection (see Theorem~\ref{theo-projection}). But in order to
  obtain a labelling that is uniform with regards to the observation
  of the quantifier, we need to make use of the narrow operation (see
  Theorem~\ref{theo-narrow}).  Intuitively, to check that a formula
  $\existsp{\cobs} \phi$ holds in a tree $\ltree$, we would like to
  work on its narrowing $\ltree'\egdef\projI[\cobs]{\ltree}$, guess a
  labelling for $p$ on this tree thanks to automata projection,
thus  obtaining a tree $\ltree'_{p}$, take its widening
  $\ltree_{p}''\egdef\liftI[{[n]}]{}{\ltree'_{p}}$, obtaining a
  tree with an $\cobs$-uniform labelling for $p$,  and then check
  that $\phi$ holds on $\ltree_{p}''$.
  The problem here is that, {while} the narrowing  $\projI[\cobs]{\tree}$
  of an unlabelled tree $\tree$
  is well defined (see Section~\ref{sec-QCTLsi-semantics}), that  of a
  labelled tree $\ltree=(\tree,\lab)$  is undefined: indeed, unless $\ltree$ is
$\obs$-uniform in every atomic proposition in $\APf$, there is no
 way
to define the labelling of $\projI[\cobs]{\tree}$ without losing
information.
This implies that we cannot feed (a narrowing of) the unfolding of the
model to our automata. Still, we need an input tree to be
successively  labelled and widened to guess uniform labellings.

\head{Splitting quantified from free propositions} To address this
problem, we remark that since we are interested in model checking a
\QCTLsi formula $\phi$ on a \CKS $\CKS$, the
automaton that we build for $\phi$ can depend on $\CKS$.  
It can thus guess paths in $\CKS$, and evaluate
free occurrences of atomic propositions in $\CKS$ without reading the
input tree. The input tree is thus no
longer used to represent the model. However we use it to carry
labellings for quantified propositions $\APq(\phi)$: we provide the
automaton with an input tree whose labelling is initially empty, and
the automaton, through successive narrowing and projection operations,
decorates it with uniform labellings for quantified atomic propositions.

We remark that this technique allows one to go beyond \CL~\cite{DBLP:conf/csl/FinkbeinerS10}: 
by separating between quantified atomic propositions (that need to be uniform) and free atomic
propositions (that state facts about the model), we manage to remove the
restriction present in \CL, that requires that all facts about
the model are known to every strategy/agent (see Proposition~\ref{prop:CLtoSLi}).

To do this we  assume without loss of generality that propositions
that are quantified in $\phi$ do not appear free in $\phi$, \ie,
$\APq(\phi)\inter\APfree(\phi)=\emptyset$. If necessary, for every
$p\in\APq(\phi)\inter\APfree(\phi)$, we take a fresh atomic
proposition $p'$ and replace all quantified occurrences of $p$ in
$\phi$ with $p'$. We obtain an equivalent formula $\phi'$ on
$\APf\,'\egdef \APf\union \{p'\mid p'\in \APq(\phi)\inter\APfree(\phi)\}$
such that $\APq(\phi')\inter\APfree(\phi')=\emptyset$.
Observe also that given a formula $\phi$ such that
$\APq(\phi)\inter\APfree(\phi)=\emptyset$, a \CKS $\CKS$ and a state $\state\in\CKS$, the
truth value of $\phi$ in $\CKS,\state$ does not depend on the
labelling of $\CKS$ for atomic propositions in $\APq(\phi)$, which can
thus be forgotten.

As a consequence, henceforth we assume that an instance $(\varphi,\CKS)$ of the
model-checking problem for \QCTLsi is such that
$\APq(\phi)\inter\APfree(\phi)=\emptyset$, and $\CKS$ is a \CKS over $\APfree(\phi)$.


\head{Merging the  decorated input tree and the model}
To state the correctness of our construction, we will need to merge
the labels for quantified propositions, carried by the input tree,
with those for free propositions, carried by \CKS $\CKS$. Because,
through successive widenings, the input tree (represented by $\ltree$
in the definition below) will necessarily be a
complete tree, its domain will always contain the domain of the
unfolding of $\CKS$ (represented by $\ltree'$ below), hence the
following definition.

\begin{definition}[Merge]
  \label{def-merge}
Let
 $\ltree=(\tree,\lab)$ be a complete
$(\APf,\Dirtree)$-tree and  $\ltree'=(\tree',\lab')$ an
$(\APf\,',\Dirtree)$-tree, where $\APf\inter\APf\,'=\emptyset$. We
 define the \emph{merge} of $\ltree$ and $\ltree'$
 as the $(\APf\union\APf\,',\Dirtree)$-tree $\ltree\merge\ltree'\egdef
(\tree\cap\tree'=\tree',\lab'')$, where
$\lab''(\noeud)=\lab(\noeud) \union \lab'(\noeud)$.
\end{definition}

We now describe our automata construction and establish the following
lemma, which is our main technical contribution.

\newcounter{lem-finalaa}
\setcounter{lem-finalaa}{\value{lemma}}

\begin{lemma}[Translation]
    \label{lem-final}
Let $(\Phi,\CKS)$ be an instance of the model-checking problem
for \QCTLsih.
    For every subformula
    $\phi$ of $\Phi$ and state $\state$ of $\CKS$, one can build an
    \ATA $\bigauto[\state]{\phi}$ on $(\APq(\Phi),\Dirtreei[\phi])$-trees
    such that for every 
    $(\APq(\Phi),\Dirtreei[\phi])$-tree $\ltree$ rooted in $\projI[{\Iphi}]{\state}$,
    \begin{equation*}
      \ltree\in\lang(\bigauto[\state]{\phi}) \mbox{\;\;\;iff\;\;\;}
      \liftI[{[n]}]{\dira}{\ltree}\merge\;\unfold{\state} \modelst \phi,
      \mbox{\;\;\; where }\dira=\projI[{[n]\setminus I_{\phi}}]{\state}.
    \end{equation*}

  \end{lemma}

For an $\Dirtreei[I]$-tree $\ltree$, from now on ${\liftI[{[n]}]{}{\ltree}}\merge
\;\unfold{\state}$ 
 stands for       ${\liftI[{[n]}]{\dira}{\ltree}}\merge
\;\unfold{\state}$, where $\dira=\projI[{[n]\setminus I}]{\state}$:
 the missing local states in the root of $\ltree$ are filled with those from $\state$.

\begin{proof}[Proof sketch of Lemma~\ref{lem-final}] 
Let $(\Phi,\CKS)$ be an instance of the model-checking problem
for \QCTLsih.
  Let $\APq=\APq(\Phi)$ and
  $\APfree=\APfree(\Phi)$, and recall that $\CKS$ is labelled over $\APfree$.
  For each state $\state\in\setstates$ and each subformula
  $\phi$ of $\Phi$ (note that all subformulas of $\Phi$ are also
  hierarchical), we define by induction on $\phi$ the \ATA
  $\bigauto{\phi}$ on $(\APq,\Dirtreei[\phi])$-trees.

~\newline
\noindent$\bm{\phi=p:}$

 We let $\bigauto{p}$ be the \ATA over $\Dirtreei[{[n]}]$-trees with one unique
  state $\tq_\init$, with transition function defined as follows:
  \[\tdelta(\tq_\init,a)=
  \begin{cases}
    \top  & \mbox{if } 
    \begin{array}{c}
p\in\APfree \mbox{ and }p\in\labS(\state)  \\
\mbox{ or   }\\ p\in \APq \mbox{ and }p\in a
    \end{array}
\\
    \perp & \mbox{if }  
    \begin{array}{c}
p\in\APfree \mbox{ and }p\notin\labS(\state) \\
 \mbox{ or  } \\
p\in \APq \mbox{ and }p\notin a      
    \end{array}
  \end{cases}
\]

~\newline
\noindent$\bm{\phi=\neg\phi':}$

We obtain $\bigauto{\phi}$  by
  dualising  $\bigauto{\phi'}$, a classic operation.

~\newline
\noindent$\bm{\phi=\phi_{1}\ou\phi_{2}:}$
  
 Because
  $\Iphi=\Iphi[\phi_{1}]\cap\Iphi[\phi_{2}]$, and each
  $\bigauto{\phi_{i}}$ works on $\Dirtreei[\phi_{i}]$-trees, we first 
  narrow them so that they work on $\Dirtreei[\phi]$-trees: 
  for $i\in \{1,2\}$, we let $\ATA_{i}\egdef
  {\narrow[\Iphi]{\bigauto{\phi_{i}}}}$.
  %
We then build $\bigauto{\phi}$ by taking the disjoint union of $\ATA_{1}$ and $\ATA_{2}$
  and adding a new initial state that nondeterministically chooses
  which of $\ATA_{1}$ or $\ATA_{2}$ to execute on the input tree, so
  that $\lang(\bigauto{\phi})=\lang(\ATA_{1})\union\lang(\ATA_{2})$.

  
~\newline
\noindent$\bm{\phi=\E\psi:}$

  The aim is to build an automaton $\bigauto{\phi}$ that works on
  $\Dirtreei[\phi]$-trees and that on input $\ltree$,  checks for the
  existence of a path in
  $\liftI[{[n]}]{}{\ltree}\merge\;\unfold{\state}$ that
  satisfies $\psi$. Observe that a path in
  $\liftI[{[n]}]{}{\ltree}\merge\;\unfold{\state}$ is a path in
  $\unfold{\state}$, augmented with the labelling for atomic
  propositions in $\APq$ carried by $\ltree$.

  To do so,  $\bigauto{\phi}$  guesses a path $\tpath$ in $(\CKS,\state)$.
It  remembers the current state in $\CKS$, which provides the
  labelling for atomic propositions in $\APfree$, and while it guesses
  $\tpath$ it follows its $\Dirtreei[\phi]$-narrowing  in its input
  tree $\ltree$ (which always exists since input to tree automata are
  complete trees),
 reading the labels  to evaluate propositions in
  $\APq$.
  
Let $\max(\psi)=\{\phi_1,\ldots,\phi_n\}$ be the
  set of maximal state sub-formulas of $\psi$.
  In a first step we
  see these maximal state sub-formulas as atomic propositions. The formula
  $\psi$
  can thus be seen as an \LTL formula, and we can build 
  a nondeterministic
  parity word automaton
  $\autopsi=(\Qpsi,\Deltapsi,\qpsi_\init,\couleurpsi)$ over alphabet $2^{\max(\psi)}$
  that accepts
  exactly the models of $\psi$~\cite{vardi1994reasoning}.
\footnote{Note that, as usual for nondeterministic word automata, we
  take the transition function of type $\Deltapsi:\Qpsi\times
  2^{\max(\psi)}\to 2^{\Qpsi}$.} 
  We define the
  \ATA
  $\tauto$  that, given as
  input a $(\max(\psi),\Dirtreei[\phi])$-tree $\ltree$,
  nondeterministically guesses a
path $\tpath$ in   $\liftI[{[n]}]{}{\ltree}\merge\;\unfold{\state}$  and
  simulates $\autopsi$ on it, assuming that the labels it reads
  while following $\projI[\Iphi]{\tpath}$
  in
  its input correctly represent the truth value of formulas in
  $\max(\psi)$ along $\tpath$. 
Recall that $\CKS=(\setstates,\relation,\state_\init,\labS)$; we define
 $\tauto\egdef(\tQ,\tdelta,\tq_{\init},\couleur)$, where
\begin{itemize}
\item $\tQ=\Qpsi\times\setstates$, 
\item $\tq_{\init}=(\qpsi_{\init},\state)$,
\item $\couleur(\qpsi,\state')=\couleurpsi(\qpsi)$, and
\item  for each $(\qpsi,\state')\in\tQ$
  and $a\in 2^{\max(\psi)}$, 
  \[\tdelta((\qpsi,\state'),a)=\bigvee_{\tq'\in\Deltapsi(\qpsi,a)}\bigvee_{
    \state''\in\relation(\state')}[\projI[\Iphi]{\state''},\left(\tq',\state''\right)].\]
\end{itemize}
The intuition is that $\tauto$ reads the current label, chooses nondeterministically
which transition to take in $\autopsi$, chooses a next state in $\CKS$
and proceeds in the corresponding direction in $\Dirtree_{\phi}$. 
Thus, $\tauto$
 accepts a $\max(\psi)$-\labeled $\Dirtreei[\phi]$-tree $\ltree$ iff
there is a path in $\ltree$ that is the $\Dirtreei[\phi]$-narrowing of
some path in $\unfold{\state}$, and 
that satisfies
 $\psi$, where maximal state formulas are considered as atomic propositions. 

Now from $\tauto$ we build the automaton $\bigauto{\phi}$ over
$\Dirtreei[\phi]$-trees labelled with ``real'' atomic propositions in
$\APq$.
In each node it visits, $\bigauto{\phi}$ guesses what should be its
labelling over $\max(\psi)$, it simulates $\tauto$
accordingly, and checks
that the guess it made is correct.
If the path being guessed in $\unfold{\state}$
is currently in node $\noeud$ ending with state $\state'$, and
$\bigauto{\phi}$ guesses that $\phi_{i}$ holds in $\noeud$,
it checks this guess by starting a
 simulation of automaton $\bigauto[\state']{\phi_{i}}$ from node
 $\noeuda=\projI[\Iphi]{\noeud}$ in its input $\ltree$.

For each $\state'\in\CKS$ and each $\phi_{i}\in\max(\psi)$ we first
build $\bigauto[\state']{\phi_i}$, and we 
 let
$\ATA^{i}_{\state'}\egdef\bigauto[\state']{\phi_{i}}=(\tQ^{i}_{\state'},\tdelta^{i}_{\state'},\tq^{i}_{\state'},\couleur^{i}_{\state'})$.
We also let
$\compl{\ATA^{i}_{\state'}}=(\compl{\tQ^{i}_{\state'}},\compl{\delta^{i}_{\state'}},\compl{\tq^{i}_{\state'}},\compl{\couleur^{i}_{\state'}})$
be its dualisation, and we assume w.l.o.g. that all the state sets are
pairwise disjoint. Observe that because each $\phi_{i}$ is a maximal
state sub-formula, we have $I_{\phi_{i}}=I_{\phi}$, so that we do not
need to narrow down automata $\ATA^{i}_{\state'}$ and
$\compl{\ATA^{i}_{\state'}}$.  We define the \ATA
\[\bigauto{\phi}=(\tQ\cup
\bigcup_{i,\state'} \tQ^{i}_{\state'} \cup
\compl{\tQ^{i}_{\state'}},\tdelta',\tq_{\init},\couleur'),\] where the
colours of states are left as they were in their original automaton,
and $\tdelta$ is defined as follows. For states in $\tQ^{i}_{\state'}$
(resp. $\compl{\tQ^{i}_{\state'}}$), $\tdelta$ agrees with $\tdelta^{i}_{\state'}$
(resp. $\compl{\delta^{i}_{\state'}}$), and for $(\qpsi,\state')\in \tQ$ and $a\in
2^{\APq}$ we let $\tdelta'((\qpsi,\state'),a)$ be the disjunction over ${a'\in
    2^{\max(\psi)}}$ of 
\begin{align*}
    \Bigg ( \tdelta\left((\qpsi,\state'),a'\right)
    \et 
      \biget_{\phi_i\in a'}\tdelta^{i}_{\state'}(\tq^{i}_{\state'},a) 
    \;\et
     \biget_{\phi_i\notin
      a'}\compl{\delta^{i}_{\state'}}(\compl{\tq^{i}_{\state'}},a)
    \Bigg ).
\end{align*}

~\newline
$\bm{\phi=\exists}^{\bm{\cobs}}\bm{p.\,\phi':}$

We  build automaton $\bigauto{\phi'}$ that works on $\Dirtreei[\phi']$-trees;
because $\phi$ is hierarchical, we have that $\cobs\subseteq I_{\phi'}$
and we can narrow down $\bigauto{\phi'}$ to work on $\Dirtreei[\cobs]$-trees and obtain
$\ATA_{1}\egdef{\narrow[{\cobs}]{\bigauto{\phi'}}}$. By
Theorem~\ref{theo-simulation} we can
nondeterminise it to get $\ATA_{2}$, which by
Theorem~\ref{theo-projection} we can project with respect to
$p$, finally obtaining $\bigauto{\phi}\egdef \proj{\ATA_{2}}$.
%
%
\end{proof}

\subsection{Proof of Theorem~\ref{theo-decidable-QCTLi}} 
We can now prove Theorem~\ref{theo-decidable-QCTLi}. Let
$\CKS$ be a \CKS, $\state \in \CKS$, and $\phi\in\QCTLsih$. By Lemma~\ref{lem-final} one can build an \ATA
$\bigauto{\phi}$ such that for every
 labelled $\Dirtreei[\phi]$-tree $\ltree$ rooted in
 $\projI[\Iphi]{\state}$, it holds that
$ \ltree\in\lang(\bigauto{\phi}) \mbox{ iff }
      \liftI[{[n]}]{}{\ltree}\merge \;\unfold{\state} \modelst
      \phi$.
      Let $\tree$ be the full $\Dirtreei[\phi]$-tree rooted in
      $\projI[\phi]{\state}$, and let
      $\ltree=(\tree,\lab_{\emptyset})$,  where $\lab_{\emptyset}$ is
       the empty labelling.
      Clearly,
$\liftI[{[n]}]{}{\ltree}\merge
\;\unfold{\state}=\unfold{\state}$, and because $\ltree$ is
rooted in $\projI[\phi]{\state}$, we have
      $ \ltree\in\lang(\bigauto{\phi}) \mbox{ iff }
      \unfold{\state}\modelst \phi$. It only remains to build a simple
       deterministic tree automaton
      $\ATA$ over $\Dirtreei[\phi]$-trees
      such that $\lang(\ATA)=\{\ltree\}$, and check for emptiness of
      the alternating tree automaton
      $\lang(\ATA\cap\bigauto{\phi})$. Because 
      nondeterminisation makes the size of the automaton gain
      one exponential for each nested quantifier on propositions, the
      procedure is nonelementary, and hardness is inherited from the
      model-checking problem for \QCTL~\cite{DBLP:journals/corr/LaroussinieM14}. 



\section{Model-checking hierarchical instances of \SLi}
\label{sec-modelcheckingSL}

In this section we establish that the model-checking problem for \SLi restricted to the class of hierarchical instances is decidable (Theorem~\ref{theo-SLi}).

We build upon the proof in
\cite{DBLP:journals/iandc/LaroussinieM15} that
establishes the decidability of the model-checking problem for \ATLssc
by reduction to the model-checking problem for \QCTLs. The main difference is that
we reduce to the model-checking problem for \QCTLsi instead, using
quantifiers on atomic propositions parameterised with observations that reflect the ones used
in the \SLi model-checking instance. 

Let $(\Phi,\CGSi)$ be a hierarchical instance of the \SLi
model-checking problem, where
$\CGSi=(\Act,\setpos,\trans,\val,\pos_\init,\obsint)$. 
We will  first show how to define a \CKS $\CKS_{\CGSi}$ and a bijection 
$\fplay\mapsto\noeud_{\fplay}$ between the set
of finite plays $\fplay$ starting in a given position $\pos$ and the set of
nodes in $\unfold[\CKS_{\CGSi}]{\state_{\pos}}$.

Then, for every subformula $\phi$ of $\Phi$ and partial
function $f:\Agf \partialto  \Varf$, we will  define a 
\QCTLsi formula $\tr[f]{\phi}$ (that will also depend on $\CGSi$) such that the following holds:
 \begin{proposition}
   \label{prop-redux}
   Suppose that $\free(\phi)\inter\Agf\subseteq\dom(f)$, and $f(a) =
   x$ implies $\assign(a) = \assign(x)$ for all $a \in \dom(f)$. Then
\[\CGSi,\assign,{\fplay}\models\phi\quad \mbox{ if and only if }\quad
\unfold[\CKS_{\CGSi}]{\state_{\fplay}} \modelst \tr[f]{\phi}.\]
 \end{proposition}

 Applying this to the sentence $\Phi$, any assignment $\assign$, and
 the empty function $\emptyset$, we get that
 \[\CGSi,\assign,\pos_{\init} \models \Phi \quad \mbox{if and only if}\quad
\unfold[\CKS_{\CGSi}]{s_{\pos_{\init}}} \models
 \tr[\emptyset]{\Phi}.\]

 \head{Constructing the \CKS $\CKS_{\CGSi}$} We will define
 $\CKS_{\CGSi}$ so that (indistinguishable) nodes in its
 tree-unfolding correspond to (indistinguishable) finite plays in
 $\CGSi$.  The \CKS will make use of atomic
 propositions $\APv\egdef\{p_{\pos}\mid\pos\in\setpos\}$ (that we
 assume to be disjoint from $\APf$).  The idea is that $p_{\pos}$
 allows the \QCTLsi formula $ \tr[\emptyset]{\Phi}$ to refer to the current position $\pos$ in
 $\CGSi$.
 Later we will see that $\tr[\emptyset]{\Phi}$ will also make use of
 atomic propositions $\APm\egdef\{p_{\mov}^{\var}\mid\mov\in\setmoves
 \mbox{ and }\var \in \Varf\}$ that we assume, again, are disjoint
 from $\APf \cup \APv$. This allows the formula to use
 $p_{\mov}^{\var}$ to refer to the actions $\mov$ {advised by
   strategies $x$. }

 Suppose $\Obsf=\{\obs_{1},\ldots,\obs_{n}\}$. For $i \in [n]$, define
 the local states
 $\setlstates_{i}\egdef\{\eqc{\obs_{i}}\mid\pos\in\setpos\}$ where
 $\eqc{\obs}$ is the equivalence class of $\pos$ for relation
 $\obseq$. Since we need to know the actual position of the \CGSi to
 define the dynamics, we also let $\setlstates_{n+1}\egdef\setpos$.

Define the \CKS $\CKS_{\CGSi}\egdef(\setstates,\relation,\state_{\init},\lab')$ where
\begin{itemize}
\item $\setstates\egdef\{\state_{\pos} \mid \pos\in\setpos\}$,
\item $\relation\egdef\{(\state_{\pos},\state_{\pos'})\mid
  \exists\jmov\in\Mov^{\Agf} \mbox{ s.t. }\trans(\pos,\jmov)=\pos'\}
  \subseteq \setstates^2$,
  \item $\state_{\init}\egdef\state_{\pos_{\init}}$,
\item $\lab'(\state_{\pos})\egdef\val(\pos)\union \{p_{\pos}\} \subseteq \APf \cup \APv$,
\end{itemize}
and $\state_{\pos}\egdef(\eqc{\obs_{1}},\ldots,\eqc{\obs_{n}},\pos) \in \prod_{i\in [n+1]}\setlstates_{i}$.


We now show how to connect finite plays in $\CGSi$ with nodes in the tree unfolding of $\CKS_{\CGSi}$.
For every finite play $\fplay=\pos_{0}\ldots\pos_{k}$, define
the node $\noeud_{\fplay}\egdef \state_{\pos_{0}}\ldots \state_{\pos_{k}}$ in
$\unfold[\CKS_{\CGSi}]{\state_{\pos_{0}}}$ (which exists, by definition of
$\CKS_{\CGSi}$ and of tree unfoldings).  Note that the mapping
$\fplay\mapsto\noeud_{\fplay}$ defines a bijection between the set
of finite plays and the set of
nodes in $\unfold[\CKS_{\CGSi}]{\state_{\init}}$.

\head{Constructing the \QCTLsih formulas $\tr[f]{\phi}$}
 We now describe how to transform an \SLi formula $\phi$ and a partial
function $f:\Agf \partialto  \Varf$ into a \QCTLsi
formula $\tr[f]{\phi}$ (that will also depend on $\CGSi$).
Suppose that $\Mov=\{\mov_{1},\ldots,\mov_{\maxmov}\}$, and define $\tr[f]{\phi}$ by induction on $\phi$. We begin with the simple cases:
$\tr[f]{p} 		 \egdef p$; $\tr[f]{\neg \phi} 	 \egdef \neg \tr[f]{\phi}$; and 
$\tr[f]{\phi_1\ou\phi_2}  \egdef \tr[f]{\phi_1}\ou\tr[f]{\phi_2}$.

We continue with the second-order quantifier case:
\begin{align*}
\tr[f]{\Estrato{\obs}\phi}	& \egdef  \exists^{\trobs{\obs}} p_{\mov_{1}}^{\var}\ldots \exists^{\trobs{\obs}} p_{\mov_{\maxmov}}^{\var}. \phistrat \et \tr[f]{\phi}
\end{align*}
where $\trobs{\obs_i} \egdef \{j\mid \obsint(\obs_{i})\subseteq\obsint(\obs_{j})\}$, and
\[
\phistrat \egdef \A\G\bigou_{\mov\in\Mov}(p_{\mov}^{\var}\et\biget_{\mov'\neq\mov}\neg p_{\mov'}^{\var}).
\]
We describe this formula in words. For each possible action $\mov\in\Mov$, an
existential quantification on the atomic proposition $p_{\mov}^{\var}$
``chooses'' for each finite play $\fplay = \pos_0 \ldots \pos_k$ of $\CGSi$
(or, equivalently, for each node $\noeud_{\fplay}$ of the tree 
$\unfold[\CKS_{\CGSi}]{\state_{{\pos_0}}}$) whether strategy $\var$
plays action $\mov$ in $\fplay$ or not. Formula $\phistrat$
ensures that  in each finite
play,  exactly one action is chosen for strategy $\var$, and thus that atomic propositions
$p_{\mov}^{\var}$ indeed characterise a
strategy, call it $\sigma_\var$.\footnote{More precisely, if $\phistrat$
  holds in node $\noeud_{\fplay}$, it ensures that propositions from $\APm$ define a
partial strategy, defined on all nodes of the subtree rooted in
$\noeud_{\fplay}$. This is enough because \SLi can only talk about the future:
when evaluating a formula in a finite play $\rho$, the definition of
strategies on plays that do not start with $\rho$ is irrelevant.}

Moreover,  a quantifier with concrete observation $\trobs{\obs_{i}}$
receives information corresponding to
observation 
$\obs_{i}$ (observe that for all $i\in [n]$, $i\in\trobs{\obs_{i}}$) as well as information
corresponding to
coarser observations.
Note that including all coarser observations does not increase the
information accessible to the quantifier: indeed, one can show that
two nodes are $\{i\}$-indistinguishable if and only if they are
$\trobs{\obs_{i}}$-indistinguishable.
However, this definition of
$\trobs{\obs_{i}}$ allows us to obtain hierarchical formulas.
Since quantification on propositions $p_{\act}^{\var}$ is done
uniformly with regards to concrete observation $\trobs{\obs_i}$, 
it follows that $\sigma_\var$ is an $\obs_i$-strategy.

Here are the remaining cases:
\begin{align*}
\tr[f]{\bind{\var}\phi}	& \egdef \tr[{f[\ag\mapsto \var]}]{\phi} 
\\
\tr[f]{\X\phi_1} 		& \egdef \A \big (\psiout[f] \impl \X\tr[f]{\phi_1} \big ) \\
\tr[f]{\phi_{1}\until\phi_{2}}	& \egdef \A \big (\psiout[f] \impl \tr[f]{\phi_{1}}\until\tr[f]{\phi_{2}} \big )
\end{align*}
where 
\[
\psiout[f]\egdef \G\left(
  \biget_{\pos\in\setpos}\biget_{\jmov\in\Mov^{\Agf}}\left ( p_{\pos} \et
  \biget_{\ag\in\Agf}p_{\jmov_{\ag}}^{f(\ag)}\impl \X
  p_{\trans(\pos,\jmov)} \right )
\right).
\]



The formula $\psiout[f]$ is  used to select the unique path assuming that every player, say $\ag$, 
follows the strategy $\sigma_{f(\ag)}$.


This completes the justification of Proposition~\ref{prop-redux}.

\head{Preserving hierarchy}
To complete the proof we show that $\tr[\emptyset]{\Phi}$ is a
hierarchical \QCTLsi formula.
This simply follows from the fact that $\Phi$ is hierarchical in
$\CGSi$ and that for every two observations $\obs_{i}$ and $\obs_{j}$ in $\Obsf$ such that
$\obsint(\obs_{i})\subseteq\obsint(\obs_{j})$, by definition of $\trobs{\obs_{k}}$
we have that $\trobs{\obs_{i}}\subseteq \trobs{\obs_{j}}$.

This completes the proof of Theorem~\ref{theo-SLi}.

\section{Outlook}

We introduced \SLi, a logic for reasoning about strategic behaviour in 
multi-player games with imperfect information.  The syntax mentions observations, and thus allows
one to write formulas that talk about dynamic observations. 
We isolated the class of hierarchical formula/model pairs $(\Phi,\CGSi)$ and proved that
one can decide if $\CGSi \models \Phi$. The proof reduces (hierarchical) instances to (hierarchical) formulas
of \QCTLsi, a low-level logic that we introduced, and that serves as a natural bridge between \SLi (that talks about 
players and strategies) and automata constructions.

We believe that
\QCTLsi is of independent interest and deserves study in its own
right. Indeed, it subsumes  MSO with equal-level predicate, which is
undecidable and of which we know no decidable fragment; yet its syntax and models
make it possible to define a natural fragment (the hierarchical
fragment) that has a simple definition,  a decidable model-checking
problem, and is suited for strategic reasoning.

Since one can alternate quantifiers in \SLi, our decidability result goes beyond synthesis. As we showed, we can use it to decide if 
a game that yields hierarchical observation has a Nash equilibrium. A crude
but easy analysis of our main decision procedure shows it is non-elementary. 

This naturally leads to a number of avenues for future work: define
and study the expressive power and computational complexity of fragments of \SLi~\cite{DBLP:conf/concur/MogaveroMPV12}; 
adapt the notion of hierarchical instances to allow for situations in which hierarchies can change infinitely often along a 
play~\cite{DBLP:conf/atva/BerwangerMB15}; and extend the logic to include
epistemic operators for individual and common knowledge, as is done in \cite{DBLP:conf/cav/CermakLMM14}, 
which are important for reasoning about distributed systems~\cite{fagin1995reasoning}.




\bibliographystyle{IEEEtran}
\bibliography{biblio}



\end{document}